\documentclass[reqno, eucal]{amsart}
\usepackage{amsmath}
\usepackage{amssymb}
\usepackage{amsfonts}
\usepackage{amsthm}
\numberwithin{equation}{section}
\allowdisplaybreaks[4]
\usepackage{cite}
\usepackage{graphicx}
\usepackage{hyperref}
\usepackage{enumerate}
\usepackage{here}
\usepackage{multirow}
\usepackage{caption}
\usepackage{lscape}
\usepackage[mathscr]{eucal}
\usepackage{braket}
\usepackage{color}
\usepackage{tikz-cd}
\usepackage{etoolbox}
\makeatletter
\patchcmd{\thmhead@plain}%
   {\thmnote{ {\the\thm@notefont(#3)}}}
   {\thmnote{  #3}}
   {}{}
\let\thmhead\thmhead@plain
\patchcmd{\swappedhead}%
   {\thmnote{ {\the\thm@notefont(#3)}}}
   {\thmnote{  #3}}
   {}{}
\let\swappedhead@plain=\swappedhead
\makeatother
\newtheorem{theorem}{Theorem}[section]
\newtheorem{lemma}[theorem]{Lemma}

\newtheorem{proposition}[theorem]{Proposition}

\theoremstyle{definition}
\newtheorem{definition}[theorem]{Definition}

\newtheorem{remark}[theorem]{Remark}
\newtheorem{notation}[theorem]{Notation}
\newcommand{\tdr}{\mathbf{R}}
\newcommand{\tdj}{\mathbf{J}}
\newcommand{\tdt}{\mathbf{T}}
\newcommand{\tdrC}{\mathcal{R}}
\newcommand{\tdmC}{\mathcal{M}}
\newcommand{\tdnC}{\mathcal{N}}
\newcommand{\tdjC}{\mathcal{J}}
\newcommand{\tdxC}{\mathcal{X}}
\newcommand{\tdtC}{\mathcal{T}}
\newcommand{\psp}{\mathbb{R}_{>0}}

\newcommand{\spaceDb}{\hspace{3em}}
\begin{document}
\title[3D reflection maps from tetrahedron maps]
{Boundary from bulk integrability in three dimensions: \\ 3D reflection maps from tetrahedron maps}
\author{Akihito Yoneyama}
\address{Akihito Yoneyama, Institute of Physics,
University of Tokyo, Komaba, Tokyo 153-8902, Japan}
\email{yoneyama@gokutan.c.u-tokyo.ac.jp}
\maketitle
\vspace{0.5cm}
\begin{center}{\bf Abstract}\end{center}
We established a method for obtaining set-theoretical solutions to the 3D reflection equation by using known ones to the Zamolodchikov tetrahedron equation, where the former equation was proposed by Isaev and Kulish as a boundary analog of the latter.
By applying our method to Sergeev's electrical solution and a two-component solution associated with the discrete modified KP equation, we obtain new solutions to the 3D reflection equation.
Our approach is closely related to a relation between the transition maps of Lusztig's parametrizations of the totally positive part of $SL_3$ and $SO_5$, which is obtained via folding the Dynkin diagram of $A_3$ into one of $B_2$.
\vspace{0.4cm}
\tableofcontents
\addtocontents{toc}{\setcounter{tocdepth}{1}}
\section{Introduction}\label{sec 1}
The Yang-Baxter equation\cite{Bax07} and reflection equation\cite{Che84,Skl88} are highly established objects describing the bulk and boundary \textit{integrability} in two dimensions, respectively.
Solutions to the former are called \textit{$R$-matrices} and it is well-known that they are systematically obtained via Drinfeld-Jimbo quantum affine algebras $U_q(\mathfrak{g})$\cite{Dri86,Jim86b}.
On the other hand, solutions to the latter consist of $R$-matrices and additional matrices called \textit{$K$-matrices}.
So far, a lot of $K$-matrices have been obtained and an analogous approach to $R$-matrices based on coideal subalgebras of $U_q(\mathfrak{g})$ is also known\cite{DM03,RV16}.
These equations are also formulated for maps on sets instead of matrices on linear spaces.
Such solutions are called \textit{set-theoretical}\cite{Dri92} and also extensively studied in connection with various topics.
See \cite{Ves07} for a review of \textit{Yang-Baxter maps} and \cite{CCZ13,CZ14,KO19,KOY05,SVW20} for examples of \textit{reflection maps}.
\par
The Zamolodchikov tetrahedron equation\cite{Zam80} and 3D reflection equation\cite{IK97} are three-dimensional analogs of the Yang-Baxter and reflection equation, respectively.
Unlike the Yang-Baxter equation, matrix solutions to the tetrahedron equation are obtained less systematically, but various solutions are derived by many reseachers\cite{BB92,BS06,CS96,KS93,KV94,SMS96,Yon20,Zam81}.
The tetrahedron equation also admits the set-theoretical formulation and such solutions have been derived: solutions to the local Yang-Baxter equation\cite{KKS98,Ser98}; transition maps of Lusztig's parametrizations of the canonical basis of quantum (super)algebras $U_q(A_2)$ and their geometric liftings, which are often called \textit{3DR}\cite{KO12,KOY13,Yon20}; solutions for relations which invariants for variants of discrete KP equations satisfy\cite{KNPT19}; and so on.
On the other hand, there are very few known solutions to the 3D reflection equation\cite{KO12,KO13,Yon20}, despite its fundamental role.
Here, all of the known set-theoretical solutions to it are transition maps of Lusztig's parametrizations of the canonical basis of quantum (super)algebras $U_q(B_2)$ and $U_q(C_2)$, and their geometric liftings.
They are direct analogs of 3DR, and often called \textit{3DJ} for $U_q(B_2)$ and \textit{3DK} for $U_q(C_2)$.
\par
\par
In this paper, we present a novel procedure for obtaining set-theoretical solutions to the 3D reflection equation from known ones to the tetrahedron equation.
In short, we derive the 3D reflection equation by cutting a composite of the tetrahedron equation into half.
\par
This idea is motivated by the two studies.
First, for some Yang-Baxter maps, it is known that the reflection equation is obtained by focusing on a half part of an identity for six products of Yang-Baxter maps, which can be derived by repeated uses of the Yang-Baxter equation\cite{CZ14,KO19}.
Here, reflection maps are obtained by cutting Yang-Baxter maps into half.
Our approach is considered as a three-dimensional analog of it.
In the construction, it is crucial that the Yang-Baxter maps satisfy a consistency condition: they output a \textit{dual} pair if they receive a dual pair as input\cite[Corollary 5]{KO19}.
Such a condition enables us to define well-defined reflection maps from Yang-Baxter maps.
Our second motivation is related to a three-dimensional analog of such a consistency condition, which is needed to define well-defined \textit{3D reflection maps} from \textit{tetrahedron maps}.
Later, we introduce the condition (\ref{bd cond}) we call \textit{boundarizable}.
This condition (\ref{bd cond}) comes from a known relation\cite{BZ01,Lus92} which represents 3DJ as a composite of 3DR, where 3DR and 3DJ are as mentioned above.
Here, the relation is obtained associated with folding the Dynkin diagram of $A_3$ into one of $B_2$.
Our condition (\ref{bd cond}) is obtained by formulating the specific result for 3DR and 3DJ as a general consistency condition.
See Remark \ref{Lus obs rem} for more details.
Actually, by applying our procedure to Sergeev's electrical solution\cite{KKS98,LP15,Ser98} and the two-component solution associated with the discrete modified KP equation\cite{KNPT19}, we can obtain new set-theoretical solutions to the 3D reflection equation.
\par
The outline of this paper is as follows.
In Section \ref{sec 20}, we introduce the notations used throughout this paper.
In Section \ref{sec 21} and \ref{sec 22}, we introduce basic definitions related to tetrahedron and 3D reflection maps.
Section \ref{sec 23} is the main part of this paper, where we introduce the \textit{boundarization} procedure for tetrahedron maps and prove that it gives 3D reflection maps.
In section \ref{sec 3}, we present examples of known tetrahedron maps satisfying our condition, and new solutions to the 3D reflection equation are presented in Section \ref{sec 32} and \ref{sec 34}.
In this paper, we mainly focus on the cases when tetrahedron and 3D reflection maps are defined on \textit{homogeneous} spaces, but this restriction is not essential.
We present a generalization for this topic in Section \ref{sec 33}.
\subsection*{Acknowledgements}
The author thanks Atsuo Kuniba and Masato Okado for their kind interest and comments.
Special thanks are due to Atsuo Kuniba for pointing out that the tetrahedron map (\ref{3dr 1para}) by \cite{LP15} is essentially equivalent to the one by \cite{Ser98}.
\addtocontents{toc}{\setcounter{tocdepth}{2}}
\section{Boundarization of tetrahedron map}\label{sec 2}
\subsection{Notation}\label{sec 20}
Throughout this paper, $X$ denotes an arbitrary set and we employ the following notation:
\begin{notation}\label{ind not}
We set the transposition $P_{ij}:X^n\to X^n$ by
\begin{align}
P_{ij}(x_1,\cdots,x_n)
=
(x_1,\cdots,x_{i-1},x_{j},x_{i+1},\cdots,x_{j-1},x_{i},x_{j+1},\cdots,x_n)
.
\end{align}
Let $\tdr:X^3\to X^3$ denote a map given by
\begin{align}
\tdr(x,y,z)=(f(x,y,z),g(x,y,z),h(x,y,z))
\quad
(f,g,h:X^3\to X)
.
\end{align}
For $i<j<k$, we define $\tdr_{ijk}:X^n\to X^n$ by
\begin{align}
\begin{split}
\tdr_{ijk}(x_1,\cdots,x_n)
=(x_1,
\cdots,&x_{i-1},f(x_i,x_j,x_k),x_{i+1}, \\
\cdots,&x_{j-1},g(x_i,x_j,x_k),x_{j+1}, \\
\cdots,&x_{k-1},h(x_i,x_j,x_k),x_{k+1},
\cdots,x_{n})
.
\end{split}
\end{align}
Otherwise, we define $\tdr_{ijk}:X^n\to X^n$ by sandwiching $\tdr$ between the permutations which sort $(i,j,k)$ in ascending order.
For example, if $i>j>k$, we define $\tdr_{ijk}:X^n\to X^n$ by $\tdr_{ijk}=P_{ik}\tdr_{kji}P_{ik}$.
More concretely, its action is given by
\begin{align}
\begin{split}
\tdr_{ijk}(x_1,\cdots,x_n)
=(x_1,
\cdots,&x_{i-1},h(x_k,x_j,x_i),x_{i+1}, \\
\cdots,&x_{j-1},g(x_k,x_j,x_i),x_{j+1}, \\
\cdots,&x_{k-1},f(x_k,x_j,x_i),x_{k+1},
\cdots,x_{n})
.
\end{split}
\end{align}
\end{notation}
\subsection{Tetrahedron equation}\label{sec 21}
In this paper, we only consider the tetrahedron equation\cite{Zam80} without spectral parameters.
We first introduce tetrahedron maps as solutions to the tetrahedron equation as with Yang-Baxter maps in two dimensions\cite{Ves07}.
\begin{definition}
Let $\tdr:X^3\to X^3$ denote a map.
We call $\tdr$ \textit{tetrahedron map} if it satisfies the \textit{tetrahedron equation} on $X^6$:
\begin{align}
\tdr_{245}
\tdr_{135}
\tdr_{126}
\tdr_{346}
=
\tdr_{346}
\tdr_{126}
\tdr_{135}
\tdr_{245}
.
\label{te}
\end{align}
\end{definition}
\begin{figure}[t]
\centering
\includegraphics[width=0.6\textwidth]{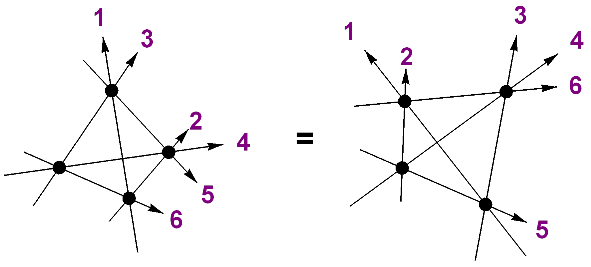}
\caption{The diagram of the tetrahedron equation (\ref{te}).}
\label{te fig}
\vspace{5mm}
\includegraphics[width=0.6\textwidth]{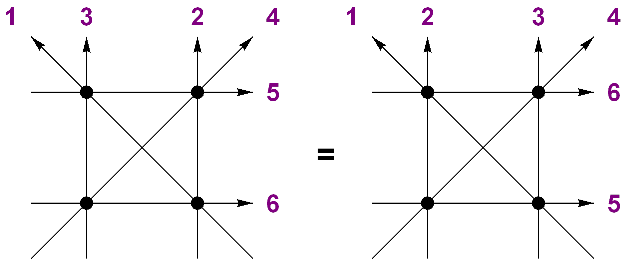}
\caption{The projection of Figure \ref{te fig} on planes.}
\label{te proj fig}
\end{figure}
The tetrahedron equation (\ref{te}) is pictorially represented as Figure \ref{te fig}, where each straight line and vertex corresponds to $X$ and $\tdr$, respectively.
Our tetrahedron equation (\ref{te}) is motivated by the reduced expressions $\mathbf{i}=(1,3,2,1,3,2)$ and $(2,1,3,2,1,3)$ for the longest element of the Weyl group of $A_3$.
See Remark \ref{Lus obs rem}.
The correspondence between our tetrahedron equation (\ref{te}) and the usual one $\tdr_{123}\tdr_{145}\tdr_{246}\tdr_{356}=\tdr_{356}\tdr_{246}\tdr_{145}\tdr_{123}$ will be presented in (\ref{te our usual}).
\par
Examples of tetrahedron maps will be given in Section \ref{sec 3}.
In this paper, we mainly focus on tetrahedron maps defined on \textit{homogeneous} spaces $X^3$ for simplicity, but Theorem \ref{main thm} can be extended to \textit{inhomogeneous} cases straightforwardly.
We will deal with such a case in Section \ref{sec 33}.
\par
From now on, we mainly focus on involutive and symmetric tetrahedron maps defined as follows:
\begin{definition}\label{invo def}
We set a tetrahedron map by $\tdr:X^3\to X^3$.
We call $\tdr$ \textit{involutive} if it satisfies $\tdr^2=\mathrm{id}$, where $\mathrm{id}$ is the identity map.
\end{definition}
\begin{definition}
We set a tetrahedron map by $\tdr:X^3\to X^3$.
We call $\tdr$ \textit{symmetric} if it satisfies $\tdr_{123}=\tdr_{321}$.
\end{definition}
These conditions are not strict.
Actually, all examples of tetrahedron maps presented in Section \ref{sec 3} are involutive, and almost all of them are symmetric except $\tdmC$ in Section \ref{sec 33}.
\par
Under these conditions, our tetrahedron equation (\ref{te}) corresponds to the usual one as follows:
\begin{align}
\begin{split}
&\tdr_{245}
\tdr_{135}
\tdr_{126}
\tdr_{346}
=
\tdr_{346}
\tdr_{126}
\tdr_{135}
\tdr_{245}
,
\\
\Longleftrightarrow
{\ }
&\tdr_{126}
\tdr_{346}
\tdr_{245}
\tdr_{135}
=
\tdr_{135}
\tdr_{245}
\tdr_{346}
\tdr_{126}
,
\\
\Longleftrightarrow
{\ }
&\tdr_{621}
\tdr_{643}
\tdr_{245}
\tdr_{135}
=
\tdr_{135}
\tdr_{245}
\tdr_{643}
\tdr_{621}
,
\end{split}
\label{te our usual}
\end{align}
where we used the involutivity for the first line and the symmetry for the second line.
Then, we obtain the usual tetrahedron equation by replacing indices as $1\to 3$, $3\to 5$, $5\to 6$ and $6\to 1$.
\subsection{3D reflection equation}\label{sec 22}
In this paper, we only consider the 3D reflection equation\cite{IK97} without spectral parameters.
We define 3D reflection maps and their involutivity as with tetrahedron maps:
\begin{definition}
Let $\tdj:X^4\to X^4$ denote a map.
We set a tetrahedron map by $\tdr:X^3\to X^3$.
We call $\tdj$ \textit{3D reflection map} if it satisfies the following \textit{3D reflection equation} on $X^9$:
\begin{align}
\tdr_{489}
\tdj_{3579}
\tdr_{269}
\tdr_{258}
\tdj_{1678}
\tdj_{1234}
\tdr_{456}
=
\tdr_{456}
\tdj_{1234}
\tdj_{1678}
\tdr_{258}
\tdr_{269}
\tdj_{3579}
\tdr_{489}
,
\label{tre}
\end{align}
where indices for $\tdj$ have the same meaning as ones given by Notation \ref{ind not}.
\end{definition}
\begin{definition}\label{invo def reflec}
We set a 3D reflection map by $\tdj:X^4\to X^4$.
We call $\tdj$ \textit{involutive} if it satisfies $\tdj^2=\mathrm{id}$.
\end{definition}
\begin{figure}[t]
\centering
\includegraphics[width=\textwidth]{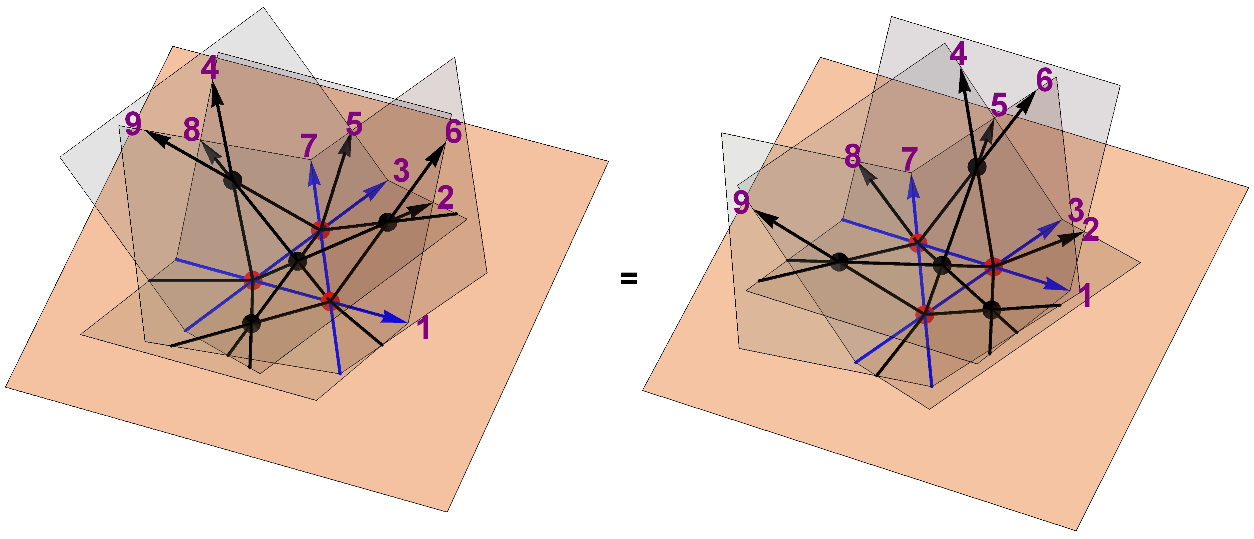}
\caption{The diagram of the 3D reflection equation (\ref{tre}).}
\label{tre fig}
\vspace{5mm}
\includegraphics[width=0.8\textwidth]{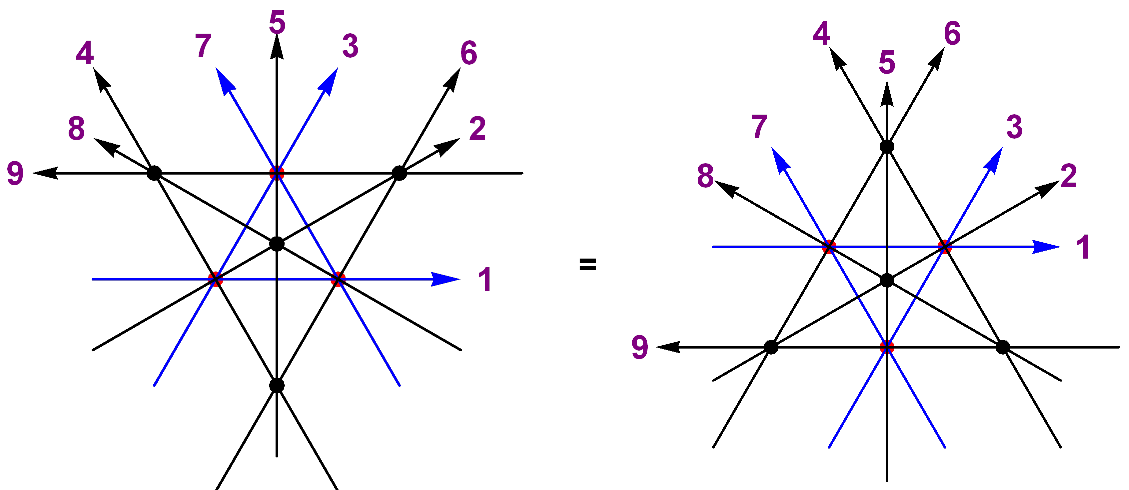}
\caption{The projection of Figure \ref{tre fig} on the boundary plane.}
\label{tre proj fig}
\end{figure}
The 3D reflection equation is pictorially represented as Figure \ref{tre fig}, where each black and red vertex corresponds to $\tdr$ and $\tdj$, respectively.
In the figure, each red vertex and blue line is on the light brown \textit{boundary plane}, and each black line hits the plane and is reflected.
It is notable that all lines in Figure \ref{tre fig} are straight.
Such a figure is obtained by putting \textit{virtual half-opened book spines} on blue lines and consider the intersections of the three books as $\tdr$, where such intersections are consequences from the elementary geometry.
In Figure \ref{tre fig}, such books are represented by translucent gray planes.
See \cite{KO13} for more detailed construction of Figure \ref{tre fig}.
\par
Examples of 3D reflection maps will be presented in Section \ref{sec 3} where all of them are involutive.
\subsection{Condition and procedure for boundarization}\label{sec 23}
We then proceed to the main part of this paper.
We first introduce some definitions for the \textit{boundarization}.
For motivations for them, see Section \ref{sec 31}.
\begin{definition}
We set a tetrahedron map by $\tdr:X^3\to X^3$.
We define a map $\tdt:X^6\to X^6$ by
\begin{align}
\tdt
=
\tdt_{123456}
=
\tdr_{245}
\tdr_{135}
\tdr_{126}
\tdr_{346}
=
\tdr_{346}
\tdr_{126}
\tdr_{135}
\tdr_{245}
.
\end{align}
We call $\tdt$ the \textit{tetrahedral composite} of the tetrahedron map $\tdr$.
\end{definition}
\begin{definition}
We define a subset of $X^6$ by
\begin{align}
Y=\{(x_1,\cdots,x_6)\mid x_2=x_3,{\ }x_5=x_6\}
.
\end{align}
We set $\phi:X^4\to Y$ and $\varphi:Y\to X^4$ by $\phi(x_1,x_2,x_3,x_4)=(x_1,x_2,x_2,x_3,x_4,x_4)$ and $\varphi(x_1,x_2,x_2,x_3,x_4,x_4)=(x_1,x_2,x_3,x_4)$, respectively.
Apparently, they are bijections and satisfy $\phi^{-1}=\varphi$ and $\varphi^{-1}=\phi$.
\end{definition}
\begin{definition}\label{bd def}
Let $\tdr:X^3\to X^3$ denote a tetrahedron map and $\tdt$ its tetrahedral composite.
We call $\tdr$ \textit{boundarizable} if the restriction of $\tdt$ to $Y$ gives a map on $Y$, that is, the following condition is satisfied:
\begin{align}
\mathbf{x}\in Y
{\ }
\Longrightarrow
{\ }
\tdt(\mathbf{x})\in Y
.
\label{bd cond}
\end{align}
\end{definition}
\begin{definition}\label{bd 3dj def}
Let $\tdr:X^3\to X^3$ denote a boundarizable tetrahedron map and $\tdt$ its tetrahedral composite.
We define the \textit{boundarization} $\tdj:X^4\to X^4$ of $\tdr$ by
\begin{align}
\tdj(\mathbf{x})
=
\varphi(\tdt(\phi(\mathbf{x})))
.
\label{bd eq}
\end{align}
\end{definition}
The boundarization given by (\ref{bd eq}) is pictorially represented as Figure \ref{bd fig}, where each black and red vertex corresponds to $\tdr$ and $\tdj$, respectively.
\begin{figure}[t]
\centering
\includegraphics[width=0.7\textwidth]{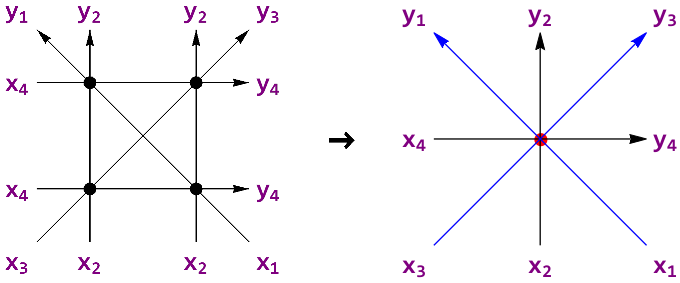}
\caption{The diagram of boundarization given by (\ref{bd eq}) where we set $y_i=\varphi(\tdt(\phi(\mathbf{x})))_{i}$ and the $i$-th component of $\varphi(\tdt(\phi(\mathbf{x})))$ by $\varphi(\tdt(\phi(\mathbf{x})))_{i}$. See also Figure \ref{te proj fig}.}
\label{bd fig}
\end{figure}
\begin{proposition}
Let $\tdr:X^3\to X^3$ denote an involutive and boundarizable tetrahedron map and $\tdj$ its boundarization.
Then, $\tdj$ is involutive.
\end{proposition}
\begin{proof}
Let $\tdt$ denote the tetrahedral composite of $\tdr$.
This is shown by the following direct calculation:
\begin{align}
\tdj^2(\mathbf{x})
=
\varphi(\tdt(\phi(\varphi(\tdt(\phi(\mathbf{x}))))))
=
\varphi(\tdt(\tdt(\phi(\mathbf{x}))))
=
\varphi(\phi(\mathbf{x}))
=
\mathbf{x}
,
\end{align}
where we used $\phi^{-1}=\varphi$, $\varphi^{-1}=\phi$ and $\tdt^2=\tdr_{245}\tdr_{135}\tdr_{126}\tdr_{346}\tdr_{346}\tdr_{126}\tdr_{135}\tdr_{245}=\mathrm{id}$ by $\tdr^2=\mathrm{id}$.
\end{proof}
The following lemma is the heart of our construction, which is a three-dimensional analog of \cite[(3.16)]{CZ14} and \cite[(5)]{KO19}.
\begin{lemma}\label{R20 lemma}
Let $\tdr$ denote an involutive and symmetric tetrahedron map.
Then, we have the following identity on $X^{15}$:
\begin{align}
\begin{split}
&(\tdr_{489}
\tdr_{\bar{4}\bar{8}\bar{9}})
(\tdr_{579}\tdr_{3\bar{5}9}\tdr_{35\bar{9}}\tdr_{\bar{5}7\bar{9}})
(\tdr_{26\bar{9}}
\tdr_{\bar{2}\bar{6}9})
(\tdr_{2\bar{5}8}
\tdr_{\bar{2}5\bar{8}})
\\
&\spaceDb
\times
(\tdr_{\bar{6}7\bar{8}}\tdr_{16\bar{8}}\tdr_{1\bar{6}8}\tdr_{678})
(\tdr_{234}\tdr_{1\bar{2}4}\tdr_{12\bar{4}}\tdr_{\bar{2}3\bar{4}})
(\tdr_{\bar{4}\bar{5}\bar{6}}
\tdr_{456})
\\
&=
(\tdr_{456}
\tdr_{\bar{4}\bar{5}\bar{6}})
(\tdr_{234}\tdr_{1\bar{2}4}\tdr_{12\bar{4}}\tdr_{\bar{2}3\bar{4}})
(\tdr_{\bar{6}7\bar{8}}\tdr_{16\bar{8}}\tdr_{1\bar{6}8}\tdr_{678})
\\
&\spaceDb
\times
(\tdr_{\bar{2}5\bar{8}}
\tdr_{2\bar{5}8})
(\tdr_{\bar{2}\bar{6}9}
\tdr_{26\bar{9}})
(\tdr_{579}\tdr_{3\bar{5}9}\tdr_{35\bar{9}}\tdr_{\bar{5}7\bar{9}})
(\tdr_{\bar{4}\bar{8}\bar{9}}
\tdr_{489})
.
\end{split}
\label{R20 eq}
\end{align}
\end{lemma}
\begin{proof}
The proof is done by simply repeated uses of the tetrahedron equation with the involutivity and symmetry for $\tdr$.
The detail of the calculation is available in Appendix \ref{app R20}, where we put the underlines to the parts to be got together or deformed.
\end{proof}
For inhomogeneous cases, not all tetrahedron maps in (\ref{R20 eq}) need to be symmetric.
See Lemma \ref{R20 super lemma} for an example.
\par
The identity (\ref{R20 eq}) is pictorially represented as Figure \ref{R20 fig}, where each vertex corresponds to $\tdr$.
In the figure, black and red vertices are above the translucent light brown plane, and brown and orange ones are below the plane.
As with Figure \ref{tre fig}, it is notable that all lines in Figure \ref{R20 fig} are straight, which is also a consequence from the elementary geometry but we omit it because it is beyond the scope of this paper.
\par
In Figure \ref{R20 fig}, brown vertices are interpreted as the \textit{mirrored image} of black ones, and red and orange vertices will be boundarized to $\tdj$.
This interpretation is theorized as follows, which is the main result of this paper.
\begin{figure}[t]
\centering
\includegraphics[width=\textwidth]{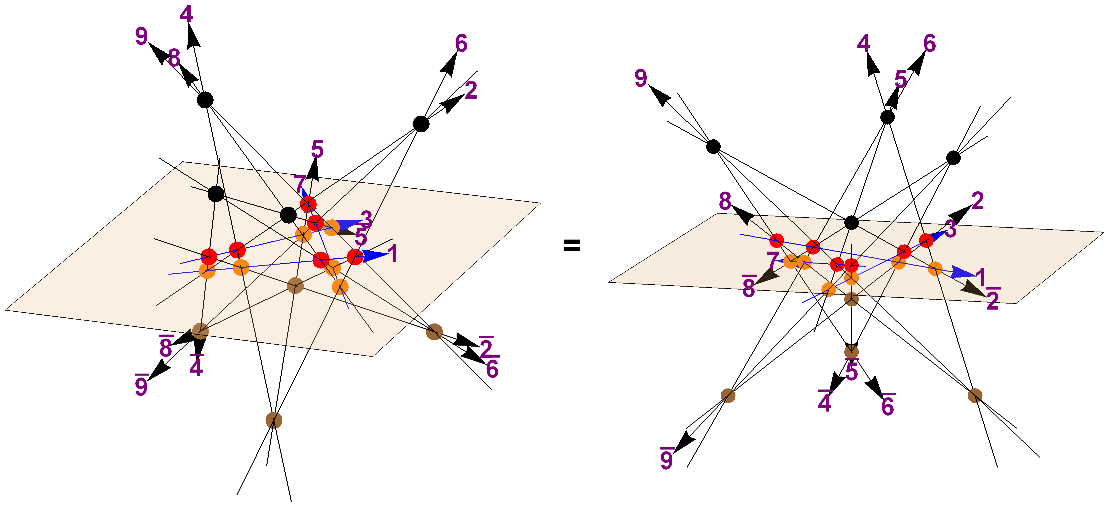}
\caption{The diagram of the identity (\ref{R20 eq}).}
\label{R20 fig}
\vspace{5mm}
\includegraphics[width=\textwidth]{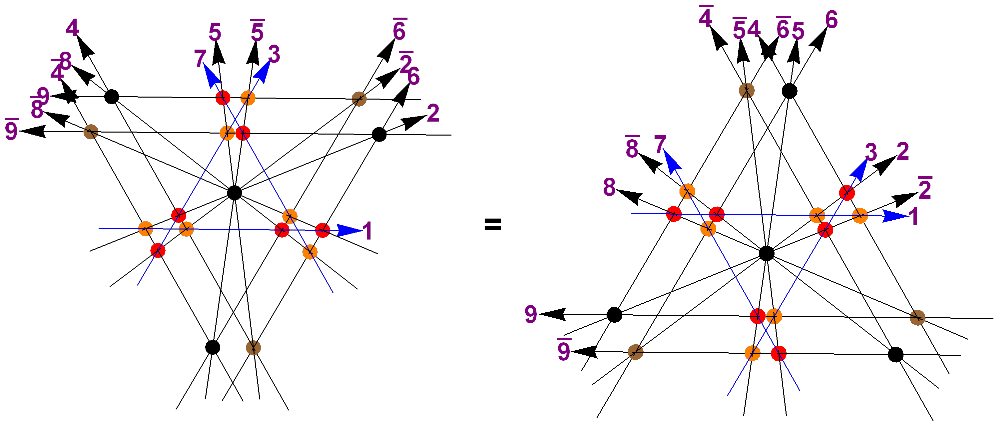}
\caption{The projection of Figure \ref{R20 fig} on the boundary plane.}
\label{R20 proj fig}
\end{figure}
\begin{theorem}\label{main thm}
Let $\tdr:X^3\to X^3$ denote an involutive, symmetric and boundarizable tetrahedron map and $\tdj:X^4\to X^4$ its boundarization.
Then, they satisfy the 3D reflection equation (\ref{tre}).
\end{theorem}
\begin{proof}
Let $\tdt$ denote the tetrahedral composite of $\tdr$.
By using $\tdt$, (\ref{R20 eq}) is written as follows:
\begin{align}
\begin{split}
&(\tdr_{489}
\tdr_{\bar{4}\bar{8}\bar{9}})
\tdt_{35\bar{5}79\bar{9}}
(\tdr_{26\bar{9}}
\tdr_{\bar{2}\bar{6}9})
(\tdr_{2\bar{5}8}
\tdr_{\bar{2}5\bar{8}})
\tdt_{16\bar{6}78\bar{8}}
\tdt_{12\bar{2}34\bar{4}}
(\tdr_{\bar{4}\bar{5}\bar{6}}
\tdr_{456})
\\
&=
(\tdr_{456}
\tdr_{\bar{4}\bar{5}\bar{6}})
\tdt_{12\bar{2}34\bar{4}}
\tdt_{16\bar{6}78\bar{8}}
(\tdr_{\bar{2}5\bar{8}}
\tdr_{2\bar{5}8})
(\tdr_{\bar{2}\bar{6}9}
\tdr_{26\bar{9}})
\tdt_{35\bar{5}79\bar{9}}
(\tdr_{\bar{4}\bar{8}\bar{9}}
\tdr_{489})
,
\end{split}
\label{R20 eq tetra}
\end{align}
where we assume that spaces are ordered as
\begin{align}
1\times 2\times \bar{2}\times 3\times 4\times \bar{4}\times 5\times \bar{5}\times 6\times \bar{6}\times 7\times 8\times \bar{8}\times 9\times \bar{9}
.
\end{align}
Let us act both sides of (\ref{R20 eq tetra}) on
\begin{align}
(x_1,x_2,x_2,x_3,x_4,x_4,x_5,x_5,x_6,x_6,x_7,x_8,x_8,x_9,x_9)
.
\end{align}
In that case, by using (\ref{bd cond}), it is easy to see that all $\tdt$ in (\ref{R20 eq tetra}) receive elements of $Y$ as inputs.
By considering (\ref{bd eq}), we then obtain the desired 3D reflection equation just viewing the upper side of the plane of Figure \ref{R20 fig}.
This is in the same way as the proof of \cite[Theorem 7]{KO19}.
\par
More formally, this can be shown by applying \textit{cutting and reconnection} procedure to (\ref{R20 eq tetra}) given by
\begin{align}
\tdt_{ij\bar{j}kl\bar{l}}
\mapsto
P_{j\bar{j}}
P_{l\bar{l}}
\tdj_{ijkl}
\tdj_{\bar{i}\bar{j}\bar{k}\bar{l}}
,
\label{cut and reconnect}
\end{align}
where we introduce new spaces $\bar{i}$ and $\bar{k}$ which are copies of the first and third spaces, respectively.
We assume that they receive the same inputs to $i$ and $k$, respectively, and the connectivity for $\bar{i}$ and $\bar{k}$ will be given as (\ref{R20 eq tetra cutted}).
The part $P_{j\bar{j}}P_{l\bar{l}}$ gives ``reconnection'', which does not change the outputs of (\ref{R20 eq tetra}) because all $\tdt$ in (\ref{R20 eq tetra}) output elements of $Y$ by (\ref{bd cond}).
The procedure (\ref{cut and reconnect}) is  pictorially represented as Figure \ref{bd proof fig}.
\begin{figure}[t]
\centering
\includegraphics[width=0.7\textwidth]{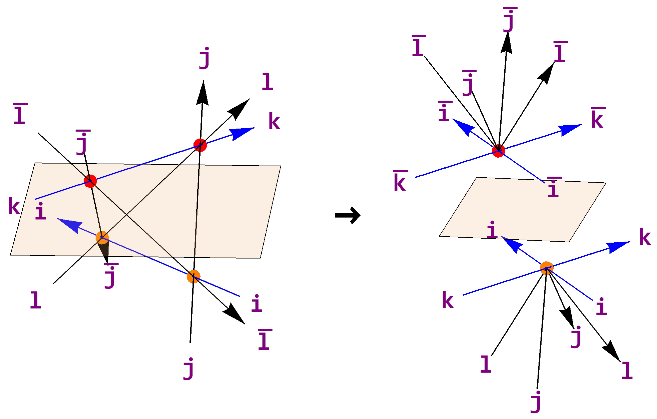}
\caption{The diagram of cutting and reconnection procedure given in the proof of Theorem \ref{main thm}.}
\label{bd proof fig}
\end{figure}
By applying (\ref{cut and reconnect}) to both sides of (\ref{R20 eq tetra}), we obtain
\begin{align}
\begin{split}
&(\tdr_{489}
\tdr_{\bar{4}\bar{8}\bar{9}})
(
P_{5\bar{5}}
P_{9\bar{9}}
\tdj_{3579}\tdj_{\bar{3}\bar{5}\bar{7}\bar{9}}
)
(\tdr_{26\bar{9}}
\tdr_{\bar{2}\bar{6}9})
(\tdr_{2\bar{5}8}
\tdr_{\bar{2}5\bar{8}})
\\
&\spaceDb
\times
(
P_{6\bar{6}}
P_{8\bar{8}}
\tdj_{1678}\tdj_{\bar{1}\bar{6}\bar{7}\bar{8}}
)
(
P_{2\bar{2}}
P_{4\bar{4}}
\tdj_{1234}\tdj_{\bar{1}\bar{2}\bar{3}\bar{4}}
)
(\tdr_{\bar{4}\bar{5}\bar{6}}
\tdr_{456})
\\
&=
(\tdr_{456}
\tdr_{\bar{4}\bar{5}\bar{6}})
(
P_{2\bar{2}}
P_{4\bar{4}}
\tdj_{1234}\tdj_{\bar{1}\bar{2}\bar{3}\bar{4}}
)
(
P_{6\bar{6}}
P_{8\bar{8}}
\tdj_{1678}\tdj_{\bar{1}\bar{6}\bar{7}\bar{8}}
)
\\
&\spaceDb
\times
(\tdr_{\bar{2}5\bar{8}}
\tdr_{2\bar{5}8})
(\tdr_{\bar{2}\bar{6}9}
\tdr_{26\bar{9}})
(
P_{5\bar{5}}
P_{9\bar{9}}
\tdj_{3579}\tdj_{\bar{3}\bar{5}\bar{7}\bar{9}}
)
(\tdr_{\bar{4}\bar{8}\bar{9}}
\tdr_{489})
.
\end{split}
\label{R20 eq tetra cutted}
\end{align}
By eliminating $P_{i\bar{i}}$, this is rewritten as follows:
\begin{align}
\begin{split}
&(\tdr_{489}
\tdr_{\bar{4}\bar{8}\bar{9}})
(
\tdj_{3579}\tdj_{\bar{3}\bar{5}\bar{7}\bar{9}}
)
(\tdr_{269}
\tdr_{\bar{2}\bar{6}\bar{9}})
(\tdr_{258}
\tdr_{\bar{2}\bar{5}\bar{8}})
(
\tdj_{1678}\tdj_{\bar{1}\bar{6}\bar{7}\bar{8}}
)
(
\tdj_{1234}\tdj_{\bar{1}\bar{2}\bar{3}\bar{4}}
)
(\tdr_{\bar{4}\bar{5}\bar{6}}
\tdr_{456})
\\
&=
(\tdr_{456}
\tdr_{\bar{4}\bar{5}\bar{6}})
(
\tdj_{1234}\tdj_{\bar{1}\bar{2}\bar{3}\bar{4}}
)
(
\tdj_{1678}\tdj_{\bar{1}\bar{6}\bar{7}\bar{8}}
)
(\tdr_{258}
\tdr_{\bar{2}\bar{5}\bar{8}})
(\tdr_{269}
\tdr_{\bar{2}\bar{6}\bar{9}})
(
\tdj_{3579}\tdj_{\bar{3}\bar{5}\bar{7}\bar{9}}
)
(\tdr_{\bar{4}\bar{8}\bar{9}}
\tdr_{489})
.
\end{split}
\label{R20 eq tetra cutted rewrite}
\end{align}
This is exactly the direct product of the 3D reflection equation.
We then obtain the desired result.
\end{proof}
\section{Examples}\label{sec 3}
\subsection{Birational and combinatorial transition map}\label{sec 31}
Let $\psp$ denote the set of positive real numbers.
We set $\tdr:\psp^3\to\psp^3$ by
\begin{align}
\tdr:
(x_1,x_2,x_3)
\mapsto
\left(
\frac{x_1x_2}{x_1+x_3},x_1+x_3,\frac{x_2x_3}{x_1+x_3}
\right)
.
\label{3dr}
\end{align}
Then, the map (\ref{3dr}) is the involutive and symmetric tetrahedron map.
This solution was obtained by \cite[(36)]{Ser98}.
Recently, this was also derived\cite[(41)]{KNPT19} by considering semi-invariants for discrete AKP equation, which appears in ABS classification as $\chi_1$\cite{ABS12}.
\par
We can verify $\tdr$ satisfies (\ref{bd cond}) by direct calculations, so it is boundarizable.
The associated boundarization $\tdj:\psp^4\to\psp^4$ is explicitly obtained as follows:
\begin{align}
\begin{split}
&\tdj:
(x_1,x_2,x_3,x_4)
\mapsto
\left(
\frac{x_1x_2^2x_3}{y_1}
,\frac{y_1}{y_2}
,\frac{y_2^2}{y_1}
,\frac{x_2x_3x_4}{y_2}
\right)
,\\
&y_1=x_1(x_2+x_4)^2+x_3x_4^2
,\quad
y_2=x_1(x_2+x_4)+x_3x_4
.
\end{split}
\label{3dj}
\end{align}
Then, $(\tdr,\tdj)$ satisfies the 3D reflection equation (\ref{tre}) by Theorem \ref{main thm}.
This is actually the known solution\cite{KO12,KO13}.
See Remark \ref{Lus obs rem} and \ref{q analog rem}.
\par
Algebraic aspects for the solution (\ref{3dr}) are deeply understood for now.
The formula (\ref{3dr}) has also appeared in \cite[Proposition 2.5]{Lus94} independently, where it was characterized as the transition map of parametrizations of the totally positive part of the special linear group $SL_3$.
Explicit formulae for such transition maps were also obtained for any semisimple Lie groups, and they all have the \textit{combinatorial} counterpart via the \textit{tropical limit}\cite{BZ01}.
For example, the map (\ref{3dr}) with $ab\to a+b$ and $a+b\to\min(a,b)$ gives the following map on $\mathbb{Z}_{\geq 0}^3$\cite[2.1]{Lus90}:
\begin{align}
\tdrC:
(x_1,x_2,x_3)
\mapsto
\left(
x_1+x_2-\min(x_1,x_3),\min(x_1,x_3),x_2+x_3-\min(x_1,x_3)
\right)
.
\label{3dr crystal}
\end{align}
This also gives the involutive and symmetric tetrahedron map.
Generally, it is known that the tropical limit of such transition maps gives the transition map of Lusztig's parametrizations of the canonical basis of the Drinfeld-Jimbo quantum algebra $U_q(\mathfrak{g})$.
The map (\ref{3dr crystal}) is the example for the case $\mathfrak{g}=sl_3$.
\par
Note that the 3D reflection map (\ref{3dj}) also admits the tropical limit, which is given by the following map on $\mathbb{Z}_{\geq 0}^4$:
\begin{align}
\begin{split}
&\tdjC:
(x_1,x_2,x_3,x_4)
\mapsto
\left(
x_1+2x_2+x_3-y_1
,y_1-y_2
,2y_2-y_1
,x_2+x_3+x_4-y_2
\right)
,\\
&y_1=\min(x_1+2\min(x_2,x_4),x_3+2x_4)
,\\
&y_2=\min(x_1+\min(x_2,x_4),x_3+x_4)
.
\end{split}
\label{3dj crystal}
\end{align}
This is also the involutive 3D reflection map, that is, $(\tdrC,\tdjC)$ satisfies the following 3D reflection equation:
\begin{align}
\tdrC_{489}
\tdjC_{3579}
\tdrC_{269}
\tdrC_{258}
\tdjC_{1678}
\tdjC_{1234}
\tdrC_{456}
=
\tdrC_{456}
\tdjC_{1234}
\tdjC_{1678}
\tdrC_{258}
\tdrC_{269}
\tdjC_{3579}
\tdrC_{489}
.
\label{tre crystal}
\end{align}
This is due to the fact that the 3D reflection map (\ref{3dj}) has a similar algebraic origin to the tetrahedron map (\ref{3dr}).
Actually, these maps (\ref{3dr}) and (\ref{3dj}) are 3DR and 3DJ referered in Section \ref{sec 1}, respevtively.
See Remark \ref{Lus obs rem}.
\begin{remark}\label{Lus obs rem}
The 3D reflection map (\ref{3dj}) is exactly the transition map of parametrizations of the totally positive part of the special orthogonal group $SO_5$ just as the tetrahedron map (\ref{3dr}) is one of $SL_3$.
This is due to the fact\cite{BZ01,Lus92} that the boundarization relation (\ref{bd eq}) for these maps is exactly the relation obtained via folding the Dynkin diagram of $A_3$ into one of $B_2$.
Here, it is crucial that the transition map for $SL_4$ is given by the tetrahedral composite of $\tdr$ like \cite[(2.44),(2.45)]{KO12}, which is associated with the reduced expressions $\mathbf{i}=(1,3,2,1,3,2)$ and $(2,1,3,2,1,3)$ for the longest element of the Weyl group of $A_3$.
This result is one of the main motivations for this study as explained in Section \ref{sec 1}.
\par
The relation between $\tdr$ and $\tdj$ was first observed by Lusztig\cite[12.5]{Lus92} and formally mentioned by \cite[Proof of Theorem 5.2]{BZ01}.
Later, Lusztig gave the detailed derivation of the above proposition in terms of Dynkin diagram automorphisms\cite[1.9]{Lus11}, and further developments are recently given by \cite{SZ20}.
\end{remark}
\begin{remark}\label{q analog rem}
The \textit{quantum} counterparts for the maps (\ref{3dr crystal}) and (\ref{3dj crystal}) are also known\cite{KV94,KO13}, which are the solutions to the tetrahedron and 3D reflection equation for matrices.
They are obtained as the intertwiners of irreducible representations of the quantum coodinate rings $A_q(SL_3)$\cite[1.12(c)]{KV94} and $A_q(SO_5)$\cite[Theorem 4.2]{KO13}, respectively.
It is remarkable that they are also characterized as the transition matrices of PBW bases of the nilpotent subalgebra of the quantum alegbras $U_q(sl_3)$ and $U_q(so_5)$, respectively\cite{KOY13}.
\par
The fact that the maps (\ref{3dr}) and (\ref{3dj}) satisfy the 3D reflection equation (\ref{tre}) was found as a by-product in the process of seeking for a geometric counterpart of such intertwiners\cite{KO12}, although the maps (\ref{3dr}) and (\ref{3dj}) themselves are known long before these developments.
\end{remark}
\subsection{Sergeev's electrical solution}\label{sec 32}
Here, we assume all variables are generic.
For $\lambda\in\mathbb{C}$, we set $\tdr(\lambda):\mathbb{C}^3\to\mathbb{C}^3$ by
\begin{align}
\tdr(\lambda):
(x_1,x_2,x_3)
\mapsto
\left(
\frac{x_1x_2}{x_1+x_3+\lambda x_1x_2x_3},x_1+x_3+\lambda x_1x_2x_3,\frac{x_2x_3}{x_1+x_3+\lambda x_1x_2x_3}
\right)
.
\label{3dr 1para}
\end{align}
Then, the map (\ref{3dr 1para}) is involutive and symmetric, and satisfies the following tetrahedron equation:
\begin{align}
\tdr(\lambda)_{245}
\tdr(\lambda)_{135}
\tdr(\lambda)_{126}
\tdr(\lambda)_{346}
=
\tdr(\lambda)_{346}
\tdr(\lambda)_{126}
\tdr(\lambda)_{135}
\tdr(\lambda)_{245}
.
\label{te ele}
\end{align}
Apparently, this solution is a one-parametric generalization of (\ref{3dr}).
Note that we can obtain (\ref{3dr 1para}) from the case $\lambda=1$ by using simple component-wise similarity transformations, so the scale of $\lambda$ does not matter.
The case for $\lambda=1$ was first obtained by \cite[(28)]{Ser98} (see also \cite{KKS98}) and later the parametric representation (\ref{3dr 1para}) was introduced by \cite[Section 4, Remark]{LP15} from an algebraic point of view.
As mentioned by \cite{Ser98}, it is associated with the electric network transformation, so often called the \textit{electrical solution}, and so on.
Interestingly, this solution was also derived\cite[(43)]{KNPT19} associated with discrete BKP equation just as (\ref{3dr}) is associated with discrete AKP equation.
\par
We can verify $\tdr(\lambda)$ satisfies (\ref{bd cond}) by dicrect calculations, so it is boundarizable.
The associated boundarization $\tdj(\lambda):\mathbb{C}^4\to\mathbb{C}^4$ is explicitly obtained as follows:
\begin{align}
\begin{split}
&\tdj(\lambda):
(x_1,x_2,x_3,x_4)
\mapsto
\left(
\frac{x_1x_2^2x_3}{y_1}
,\frac{y_1}{y_2}
,\frac{y_2^2}{y_1}
,\frac{x_2x_3x_4}{y_2}
\right)
,\\
&y_1=x_1(x_2+x_4)(x_2+x_4+2\lambda x_2x_3x_4)+x_3x_4^2
,\\
&y_2=x_1(x_2+x_4+2\lambda x_2x_3x_4)+x_3x_4
.
\end{split}
\label{3dj 1para}
\end{align}
Similarly to (\ref{3dr 1para}), this gives (\ref{3dj}) when we set $\lambda=0$.
We note that the formula (\ref{3dj 1para}) appears in \cite[Section 5, Remark]{LP15} without any derivations.
We hope that our result gives a new insight into the topic about electrical Lie groups.
By Theorem \ref{main thm}, $(\tdr(\lambda),\tdj(\lambda))$ gives a new solution to the 3D reflection equation, which is a one-parametric generalization of one of Section \ref{sec 31}:
\begin{proposition}
\begin{align}
\begin{split}
&\tdr(\lambda)_{489}
\tdj(\lambda)_{3579}
\tdr(\lambda)_{269}
\tdr(\lambda)_{258}
\tdj(\lambda)_{1678}
\tdj(\lambda)_{1234}
\tdr(\lambda)_{456}
\\
&=
\tdr(\lambda)_{456}
\tdj(\lambda)_{1234}
\tdj(\lambda)_{1678}
\tdr(\lambda)_{258}
\tdr(\lambda)_{269}
\tdj(\lambda)_{3579}
\tdr(\lambda)_{489}
.
\end{split}
\label{tre 1para}
\end{align}
\end{proposition}
\vspace{0mm}
\subsection{Super analog of combinatorial transition map}\label{sec 33}
In this section, we present an inhomogeneous extension of Theorem \ref{main thm}.
We set $\tdmC:\mathbb{Z}_{\geq 0}\times\{0,1\}\times\{0,1\}\to\mathbb{Z}_{\geq 0}\times\{0,1\}\times\{0,1\}$ and $\tdnC:\{0,1\}\times\mathbb{Z}_{\geq 0}\times\{0,1\}\to\{0,1\}\times\mathbb{Z}_{\geq 0}\times\{0,1\}$ by
\begin{align}
&\tdmC:
(x_1,x_2,x_3)
\mapsto
\left(
x_1+x_2-v
,v
,x_2+x_3-v
\right)
\label{3dm}
,\\
&\tdnC:
(x_1,x_2,x_3)
\mapsto
\left(
x_1+w
,x_2-w
,x_3+w
\right)
,
\label{3dn}
\end{align}
where we set $v=\min(x_1+x_2,x_3)$ and $w=\min(x_2,1-x_1-x_3)$.
Then, both of them are involutive and the map (\ref{3dn}) is also symmetric.
They satisfy the following tetrahedron equations on $\mathbb{Z}_{\geq 0}\times\{0,1\}^2\times\mathbb{Z}_{\geq 0}\times\{0,1\}^2$ and $\mathbb{Z}_{\geq 0}^3\times\{0,1\}^3$, respectively:
\begin{lemma}
\begin{align}
\tdnC_{245}
\tdmC_{135}
\tdmC_{126}
\tdnC_{346}
&=
\tdnC_{346}
\tdmC_{126}
\tdmC_{135}
\tdnC_{245}
,
\label{te super 1}
\\
\tdrC_{123}
\tdmC_{145}
\tdmC_{246}
\tdmC_{356}
&=
\tdmC_{356}
\tdmC_{246}
\tdmC_{145}
\tdrC_{123}
,
\label{te super 2}
\end{align}
where $\tdrC$ is given by (\ref{3dr crystal}).
\end{lemma}
\begin{proof}
The second equation (\ref{te super 2}) is exactly \cite[Corollary 6.6]{Yon20}.
The first equation (\ref{te super 1}) can be proved by exhaustion on values of the subspace $\{0,1\}^4$.
For example, if we apply the both sides of (\ref{te super 1}) to $(x_1,0,0,x_4,0,0){\ }(x_1,x_4\in\mathbb{Z}_{\geq 0})$, they are calculated as follows:
\begin{enumerate}[(1)]
\setlength{\leftskip}{1em}
\item
\par
The case $x_1,x_4\geq 1$.
\begin{align}
\begin{split}
(x_1,0,0,x_4,0,0)
&\overset{\tdnC_{346}}{\mapsto}
(x_1,0,1,x_4-1,0,1)
\overset{\tdmC_{126}}{\mapsto}
(x_1-1,1,1,x_4-1,0,0)
\\
&\overset{\tdmC_{135}}{\mapsto}
(x_1,1,0,x_4-1,1,0)
\overset{\tdnC_{245}}{\mapsto}
(x_1,0,0,x_4,0,0)
,
\end{split}
\label{te super 1 proof 1 case1}
\\
\begin{split}
(x_1,0,0,x_4,0,0)
&\overset{\tdnC_{245}}{\mapsto}
(x_1,1,0,x_4-1,1,0)
\overset{\tdmC_{135}}{\mapsto}
(x_1-1,1,1,x_4-1,0,0)
\\
&\overset{\tdmC_{126}}{\mapsto}
(x_1,0,1,x_4-1,0,1)
\overset{\tdnC_{346}}{\mapsto}
(x_1,0,0,x_4,0,0)
.
\end{split}
\label{te super 1 proof 2 case1}
\end{align}
\item
\par
The case $x_1=0$ and $x_4\geq 1$.
\begin{align}
\begin{split}
(0,0,0,x_4,0,0)
&\overset{\tdnC_{346}}{\mapsto}
(0,0,1,x_4-1,0,1)
\overset{\tdmC_{126}}{\mapsto}
(0,0,1,x_4-1,0,1)
\\
&\overset{\tdmC_{135}}{\mapsto}
(1,0,0,x_4-1,1,1)
\overset{\tdnC_{245}}{\mapsto}
(1,0,0,x_4-1,1,1)
,
\end{split}
\label{te super 1 proof 1 case2}
\\
\begin{split}
(0,0,0,x_4,0,0)
&\overset{\tdnC_{245}}{\mapsto}
(0,1,0,x_4-1,1,0)
\overset{\tdmC_{135}}{\mapsto}
(0,1,0,x_4-1,1,0)
\\
&\overset{\tdmC_{126}}{\mapsto}
(1,0,0,x_4-1,1,1)
\overset{\tdnC_{346}}{\mapsto}
(1,0,0,x_4-1,1,1)
.
\end{split}
\label{te super 1 proof 2 case2}
\end{align}
\item
\par
The case $x_1\geq 0$ and $x_4=0$.
\begin{align}
\begin{split}
(x_1,0,0,0,0,0)
&\overset{\tdnC_{346}}{\mapsto}
(x_1,0,0,0,0,0)
\overset{\tdmC_{126}}{\mapsto}
(x_1,0,0,0,0,0)
\\
&\overset{\tdmC_{135}}{\mapsto}
(x_1,0,0,0,0,0)
\overset{\tdnC_{245}}{\mapsto}
(x_1,0,0,0,0,0)
,
\end{split}
\label{te super 1 proof 1 case3}
\\
\begin{split}
(x_1,0,0,0,0,0)
&\overset{\tdnC_{245}}{\mapsto}
(x_1,0,0,0,0,0)
\overset{\tdmC_{135}}{\mapsto}
(x_1,0,0,0,0,0)
\\
&\overset{\tdmC_{126}}{\mapsto}
(x_1,0,0,0,0,0)
\overset{\tdnC_{346}}{\mapsto}
(x_1,0,0,0,0,0)
.
\end{split}
\label{te super 1 proof 2 case3}
\end{align}
\end{enumerate}
They give (\ref{te super 1}) for the case $(0,0,0,0)\in\{0,1\}^4$.
\end{proof}
The maps (\ref{3dm}) and (\ref{3dn}) were obtained by considering the crystal limit of the transition matrices of PBW bases of the nilpotent subalgebra of the quantum superalgebras of type A\cite[Sec 6.1]{Yon20}, so they are natural analogs of the tetrahedron map (\ref{3dr crystal}).
\begin{lemma}
The tetrahedral composite $\tdtC=\tdnC_{245}\tdmC_{135}\tdmC_{126}\tdnC_{346}$ on $\mathbb{Z}_{\geq 0}\times\{0,1\}^2\times\mathbb{Z}_{\geq 0}\times\{0,1\}^2$ satisfies the boundarization condition (\ref{bd cond}).
\end{lemma}
\begin{proof}
This can be proved by exhaustion.
For example, (\ref{te super 1 proof 1 case1}), (\ref{te super 1 proof 1 case2}) and (\ref{te super 1 proof 1 case3}) verify the case $x_2=x_3=x_5=x_6=0$.
\end{proof}
\par
Let $\tdxC$ denote a map on $\mathbb{Z}_{\geq 0}\times\{0,1\}\times\mathbb{Z}_{\geq 0}\times\{0,1\}$ given by
\begin{align}
\begin{split}
&\tdxC:
(x_1,x_2,x_3,x_4)
\mapsto
\left(
x_1+2x_2+x_3-y_1
,y_1-y_2
,2y_2-y_1
,x_2+x_3+x_4-y_2
\right)
,\\
&y_1=\max(x_3+\min(x_1+2x_2+x_4-1,2x_4),0)
,\\
&y_2=\max(x_3+\min(x_1+x_2+x_4-1,x_4),0)
.
\end{split}
\label{3dx}
\end{align}
The map (\ref{3dx}) was obtained by \cite[Sec 6.1]{Yon20}.
The origin of it is similar to the maps (\ref{3dm}) and (\ref{3dn}): it is obtained by considering the crystal limit of the transition matrix for the quantum superalgebras of type B, so a natural analog of the 3D reflection map (\ref{3dj crystal}).
As shown in \cite{Yon20}, the map (\ref{3dx}) is involutive.
\par
We have a super analog of the relation which mentioned in Remark \ref{Lus obs rem} by exhaustion:
\begin{proposition}
The boundarization associated with the tetrahedral composite $\tdtC=\tdnC_{245}\tdmC_{135}\tdmC_{126}\tdnC_{346}$ is given by $\tdxC$.
\end{proposition}
\par
Then, we present an inhomogeneous extension of Lemma \ref{R20 lemma}.
\begin{lemma}\label{R20 super lemma}
\begin{align}
\begin{split}
&(\tdmC_{489}
\tdmC_{\bar{4}\bar{8}\bar{9}})
(\tdnC_{579}\tdmC_{3\bar{5}9}\tdmC_{35\bar{9}}\tdnC_{\bar{5}7\bar{9}})
(\tdmC_{26\bar{9}}
\tdmC_{\bar{2}\bar{6}9})
(\tdmC_{2\bar{5}8}
\tdmC_{\bar{2}5\bar{8}})
\\
&\spaceDb
\times
(\tdnC_{\bar{6}7\bar{8}}\tdmC_{16\bar{8}}\tdmC_{1\bar{6}8}\tdnC_{678})
(\tdrC_{234}\tdrC_{1\bar{2}4}\tdrC_{12\bar{4}}\tdrC_{\bar{2}3\bar{4}})
(\tdmC_{\bar{4}\bar{5}\bar{6}}
\tdmC_{456})
\\
&=
(\tdmC_{456}
\tdmC_{\bar{4}\bar{5}\bar{6}})
(\tdrC_{234}\tdrC_{1\bar{2}4}\tdrC_{12\bar{4}}\tdrC_{\bar{2}3\bar{4}})
(\tdnC_{\bar{6}7\bar{8}}\tdmC_{16\bar{8}}\tdmC_{1\bar{6}8}\tdnC_{678})
\\
&\spaceDb
\times
(\tdmC_{\bar{2}5\bar{8}}
\tdmC_{2\bar{5}8})
(\tdmC_{\bar{2}\bar{6}9}
\tdmC_{26\bar{9}})
(\tdnC_{579}\tdmC_{3\bar{5}9}\tdmC_{35\bar{9}}\tdnC_{\bar{5}7\bar{9}})
(\tdmC_{\bar{4}\bar{8}\bar{9}}
\tdmC_{489})
,
\end{split}
\label{R20 super eq}
\end{align}
where $\tdrC$ is given by (\ref{3dr crystal}).
\end{lemma}
\begin{proof}
This is shown in the same way as Lemma \ref{R20 lemma}.
The detail of the calculation is available in Appendix \ref{app R20 super}, where we put the underlines to the parts to be got together or deformed.
As we can see by Appendix \ref{app R20 super}, we do not need the symmetry for $\tdmC$.
\end{proof}
\par
By using the above lemmas, proposition and the properties for $(\tdrC,\tdjC)$, we obtain the following solution to the 3D reflection equation:
\begin{theorem}\label{main thm super}
\begin{align}
\tdmC_{489}
\tdxC_{3579}
\tdmC_{269}
\tdmC_{258}
\tdxC_{1678}
\tdjC_{1234}
\tdmC_{456}
=
\tdmC_{456}
\tdjC_{1234}
\tdxC_{1678}
\tdmC_{258}
\tdmC_{269}
\tdxC_{3579}
\tdmC_{489}
.
\label{my tre crys}
\end{align}
\end{theorem}
\begin{proof}
This is shown in the same way as Theorem \ref{main thm}.
\end{proof}
Note that the solution (\ref{my tre crys}) has already obtained\cite[(6.33)]{Yon20}, but our construction gives an alternative derivation for this solution.
\subsection{Two-component solution associated with soliton equation}\label{sec 34}
Here, we assume all variables are generic.
We set $\tdr:(\mathbb{C}^2)^3\to(\mathbb{C}^2)^3$ by
\begin{align}
\tdr:
\left(
\begin{pmatrix}
x_1
\\
y_1
\end{pmatrix}
,
\begin{pmatrix}
x_2
\\
y_2
\end{pmatrix}
,
\begin{pmatrix}
x_3
\\
y_3
\end{pmatrix}
\right)
\mapsto
\left(
\begin{pmatrix}
\displaystyle
\frac{x_1x_2}{x_1+x_3}
\\
\\
\displaystyle
\frac{(x_1+x_3)y_1y_2}{x_1y_1+x_3y_3}
\end{pmatrix}
,
\begin{pmatrix}
\displaystyle
x_1+x_3
\\
\\
\displaystyle
\frac{x_1y_1+x_3y_3}{x_1+x_3}
\end{pmatrix}
,
\begin{pmatrix}
\displaystyle
\frac{x_2x_3}{x_1+x_3}
\\
\\
\displaystyle
\frac{(x_1+x_3)y_2y_3}{x_1y_1+x_3y_3}
\end{pmatrix}
\right)
.
\label{3dr vec}
\end{align}
Then, the map (\ref{3dr vec}) is the involutive and symmetric tetrahedron map.
This solution was obtained\cite[(93)]{KNPT19} associated with discrete modified KP equation, which appears in ABS classification as $\chi_4$\cite{ABS12}.
This solution is a two-component generalization of (\ref{3dr}) because this gives (\ref{3dr}) when we set $y_i=1{\ }(i=1,2,3)$.
\par
We can verify $\tdr$ satisfies (\ref{bd cond}) by direct calculations, so it is boundarizable.
The associated boundarization $\tdj:(\mathbb{C}^2)^4\to (\mathbb{C}^2)^4$ is explicitly obtained as follows:
\begin{align}
\begin{split}
&\tdj:
\left(
\begin{pmatrix}
x_1
\\
y_1
\end{pmatrix}
,
\begin{pmatrix}
x_2
\\
y_2
\end{pmatrix}
,
\begin{pmatrix}
x_3
\\
y_3
\end{pmatrix}
,
\begin{pmatrix}
x_4
\\
y_4
\end{pmatrix}
\right)
\mapsto
\left(
\begin{pmatrix}
\displaystyle
\frac{x_1x_2^2x_3}{z_1}
\\
\\
\displaystyle
\frac{y_1y_2^2y_3z_1}{w_1}
\end{pmatrix}
,
\begin{pmatrix}
\displaystyle
\frac{z_1}{z_2}
\\
\\
\displaystyle
\frac{z_2w_1}{z_1w_2}
\end{pmatrix}
,
\begin{pmatrix}
\displaystyle
\frac{z_2^2}{z_1}
\\
\\
\displaystyle
\frac{z_1w_2^2}{z_2^2w_1}
\end{pmatrix}
,
\begin{pmatrix}
\displaystyle
\frac{x_2x_3x_4}{z_2}
\\
\\
\displaystyle
\frac{y_2y_3y_4z_2}{w_2}
\end{pmatrix}
\right)
,\\
&z_1=x_1(x_2+x_4)^2+x_3x_4^2
,\quad
z_2=x_1(x_2+x_4)+x_3x_4
,\\
&w_1=x_1y_1(x_2y_2+x_4y_4)^2+x_3x_4^2y_3y_4^2
,\quad
w_2=x_1y_1(x_2y_2+x_4y_4)+x_3x_4y_3y_4
.
\end{split}
\label{3dj vec}
\end{align}
Similarly to (\ref{3dr vec}), this gives (\ref{3dj}) when we set $y_i=1{\ }(i=1,2,3,4)$.
By Theorem \ref{main thm}, $(\tdr,\tdj)$ gives a new solution to the 3D reflection equation, which is a two-component generalization of one of Section \ref{sec 31}:
\begin{proposition}
$(\tdr,\tdj)$ given by (\ref{3dr vec}) and (\ref{3dj vec}) satisfies the 3D reflection equation (\ref{tre}).
\end{proposition}
\vspace{0mm}
\appendix
\begin{landscape}
\section{Proof of Lemma \ref{R20 lemma}}\label{app R20}
\begin{align}
&
\tdr_{489}
\tdr_{\bar{4}\bar{8}\bar{9}}
\underline{
\tdr_{579}\tdr_{3\bar{5}9}\tdr_{35\bar{9}}\tdr_{\bar{5}7\bar{9}}
}
\tdr_{26\bar{9}}
\tdr_{\bar{2}\bar{6}9}
\tdr_{2\bar{5}8}
\tdr_{\bar{2}5\bar{8}}
\tdr_{\bar{6}7\bar{8}}
\tdr_{16\bar{8}}\tdr_{1\bar{6}8}\tdr_{678}
\tdr_{234}\tdr_{1\bar{2}4}\tdr_{12\bar{4}}\tdr_{\bar{2}3\bar{4}}
\tdr_{\bar{4}\bar{5}\bar{6}}
\tdr_{456}
\\
&=
\tdr_{489}
\tdr_{\bar{4}\bar{8}\bar{9}}
\tdr_{\bar{5}7\bar{9}}\tdr_{35\bar{9}}\tdr_{3\bar{5}9}
\underline{\tdr_{579}}%
\tdr_{26\bar{9}}
\underline{\tdr_{\bar{2}\bar{6}9}}%
\tdr_{2\bar{5}8}
\underline{\tdr_{\bar{2}5\bar{8}}}%
{\ }
\underline{\tdr_{\bar{6}7\bar{8}}}%
\tdr_{16\bar{8}}\tdr_{1\bar{6}8}\tdr_{678}
\tdr_{234}\tdr_{1\bar{2}4}\tdr_{12\bar{4}}\tdr_{\bar{2}3\bar{4}}
\tdr_{\bar{4}\bar{5}\bar{6}}
\tdr_{456}
\\
&=
\tdr_{489}
\tdr_{\bar{4}\bar{8}\bar{9}}
\tdr_{\bar{5}7\bar{9}}\tdr_{35\bar{9}}\tdr_{3\bar{5}9}
\tdr_{26\bar{9}}
\tdr_{2\bar{5}8}
\underline{
\tdr_{579}
\tdr_{\bar{2}\bar{6}9}
\tdr_{\bar{2}5\bar{8}}
\tdr_{\bar{6}7\bar{8}}
}%
\tdr_{16\bar{8}}\tdr_{1\bar{6}8}\tdr_{678}
\tdr_{234}\tdr_{1\bar{2}4}\tdr_{12\bar{4}}\tdr_{\bar{2}3\bar{4}}
\tdr_{\bar{4}\bar{5}\bar{6}}
\tdr_{456}
\\
&=
\underline{\tdr_{489}}%
\tdr_{\bar{4}\bar{8}\bar{9}}
\tdr_{\bar{5}7\bar{9}}\tdr_{35\bar{9}}
\underline{\tdr_{3\bar{5}9}}%
\tdr_{26\bar{9}}
\underline{\tdr_{2\bar{5}8}}%
\tdr_{\bar{6}7\bar{8}}
\tdr_{\bar{2}5\bar{8}}
\tdr_{\bar{2}\bar{6}9}
\tdr_{579}
\tdr_{16\bar{8}}\tdr_{1\bar{6}8}\tdr_{678}
\underline{\tdr_{234}}%
\tdr_{1\bar{2}4}\tdr_{12\bar{4}}\tdr_{\bar{2}3\bar{4}}
\tdr_{\bar{4}\bar{5}\bar{6}}
\tdr_{456}
\\
&=
\tdr_{\bar{4}\bar{8}\bar{9}}
\tdr_{\bar{5}7\bar{9}}\tdr_{35\bar{9}}
\tdr_{26\bar{9}}
\tdr_{\bar{6}7\bar{8}}
\tdr_{\bar{2}5\bar{8}}
\underline{
\tdr_{489}
\tdr_{3\bar{5}9}
\tdr_{2\bar{5}8}
\tdr_{234}
}
\tdr_{\bar{2}\bar{6}9}
\tdr_{579}
\tdr_{16\bar{8}}\tdr_{1\bar{6}8}\tdr_{678}
\tdr_{1\bar{2}4}\tdr_{12\bar{4}}\tdr_{\bar{2}3\bar{4}}
\tdr_{\bar{4}\bar{5}\bar{6}}
\tdr_{456}
\\
&=
\tdr_{\bar{4}\bar{8}\bar{9}}
\tdr_{\bar{5}7\bar{9}}\tdr_{35\bar{9}}
\tdr_{26\bar{9}}
\tdr_{\bar{6}7\bar{8}}
\tdr_{\bar{2}5\bar{8}}
\tdr_{234}
\tdr_{2\bar{5}8}
\tdr_{3\bar{5}9}
\underline{\tdr_{489}}%
{\ }
\underline{\tdr_{\bar{2}\bar{6}9}}%
\tdr_{579}
\tdr_{16\bar{8}}
\underline{\tdr_{1\bar{6}8}}%
\tdr_{678}
\underline{\tdr_{1\bar{2}4}}%
\tdr_{12\bar{4}}\tdr_{\bar{2}3\bar{4}}
\tdr_{\bar{4}\bar{5}\bar{6}}
\tdr_{456}
\\
&=
\tdr_{\bar{4}\bar{8}\bar{9}}
\tdr_{\bar{5}7\bar{9}}\tdr_{35\bar{9}}
\tdr_{26\bar{9}}
\tdr_{\bar{6}7\bar{8}}
\tdr_{\bar{2}5\bar{8}}
\tdr_{234}
\tdr_{2\bar{5}8}
\tdr_{3\bar{5}9}
\tdr_{16\bar{8}}
\underline{
\tdr_{489}
\tdr_{\bar{2}\bar{6}9}
\tdr_{1\bar{6}8}
\tdr_{1\bar{2}4}
}
\tdr_{579}
\tdr_{678}
\tdr_{12\bar{4}}\tdr_{\bar{2}3\bar{4}}
\tdr_{\bar{4}\bar{5}\bar{6}}
\tdr_{456}
\\
&=
\tdr_{\bar{4}\bar{8}\bar{9}}
\tdr_{\bar{5}7\bar{9}}\tdr_{35\bar{9}}
\tdr_{26\bar{9}}
\tdr_{\bar{6}7\bar{8}}
\tdr_{\bar{2}5\bar{8}}
\tdr_{234}
\tdr_{2\bar{5}8}
\tdr_{3\bar{5}9}
\tdr_{16\bar{8}}
\tdr_{1\bar{2}4}
\tdr_{1\bar{6}8}
\tdr_{\bar{2}\bar{6}9}
\underline{\tdr_{489}}%
{\ }
\underline{\tdr_{579}}%
{\ }
\underline{\tdr_{678}}%
\tdr_{12\bar{4}}\tdr_{\bar{2}3\bar{4}}
\tdr_{\bar{4}\bar{5}\bar{6}}
\underline{\tdr_{456}}
\\
&=
\tdr_{\bar{4}\bar{8}\bar{9}}
\tdr_{\bar{5}7\bar{9}}\tdr_{35\bar{9}}
\tdr_{26\bar{9}}
\tdr_{\bar{6}7\bar{8}}
\tdr_{\bar{2}5\bar{8}}
\tdr_{234}
\tdr_{2\bar{5}8}
\tdr_{3\bar{5}9}
\tdr_{16\bar{8}}
\tdr_{1\bar{2}4}
\tdr_{1\bar{6}8}
\tdr_{\bar{2}\bar{6}9}
\tdr_{12\bar{4}}\tdr_{\bar{2}3\bar{4}}
\tdr_{\bar{4}\bar{5}\bar{6}}
\underline{
\tdr_{489}
\tdr_{579}
\tdr_{678}
\tdr_{456}
}
\\
&=
\tdr_{\bar{4}\bar{8}\bar{9}}
\tdr_{\bar{5}7\bar{9}}\tdr_{35\bar{9}}
\tdr_{26\bar{9}}
\tdr_{\bar{6}7\bar{8}}
\tdr_{\bar{2}5\bar{8}}
\tdr_{234}
\tdr_{2\bar{5}8}
\underline{\tdr_{3\bar{5}9}}%
\tdr_{16\bar{8}}
\tdr_{1\bar{2}4}
\tdr_{1\bar{6}8}
\underline{\tdr_{\bar{2}\bar{6}9}}%
\tdr_{12\bar{4}}
\underline{\tdr_{\bar{2}3\bar{4}}}%
{\ }
\underline{\tdr_{\bar{4}\bar{5}\bar{6}}}%
\tdr_{456}
\tdr_{678}
\tdr_{579}
\tdr_{489}
\\
&=
\tdr_{\bar{4}\bar{8}\bar{9}}
\tdr_{\bar{5}7\bar{9}}\tdr_{35\bar{9}}
\tdr_{26\bar{9}}
\tdr_{\bar{6}7\bar{8}}
\tdr_{\bar{2}5\bar{8}}
\tdr_{234}
\tdr_{2\bar{5}8}
\tdr_{16\bar{8}}
\tdr_{1\bar{2}4}
\tdr_{1\bar{6}8}
\tdr_{12\bar{4}}
\underline{
\tdr_{3\bar{5}9}
\tdr_{\bar{2}\bar{6}9}
\tdr_{\bar{2}3\bar{4}}
\tdr_{\bar{4}\bar{5}\bar{6}}
}
\tdr_{456}
\tdr_{678}
\tdr_{579}
\tdr_{489}
\\
&=
\tdr_{\bar{4}\bar{8}\bar{9}}
\tdr_{\bar{5}7\bar{9}}\tdr_{35\bar{9}}
\tdr_{26\bar{9}}
\tdr_{\bar{6}7\bar{8}}
\tdr_{\bar{2}5\bar{8}}
\tdr_{234}
\underline{\tdr_{2\bar{5}8}}%
\tdr_{16\bar{8}}
\tdr_{1\bar{2}4}
\underline{\tdr_{1\bar{6}8}}%
{\ }
\underline{\tdr_{12\bar{4}}}%
{\ }
\underline{\tdr_{\bar{4}\bar{5}\bar{6}}}%
\tdr_{\bar{2}3\bar{4}}
\tdr_{\bar{2}\bar{6}9}
\tdr_{3\bar{5}9}
\tdr_{456}
\tdr_{678}
\tdr_{579}
\tdr_{489}
\\
&=
\tdr_{\bar{4}\bar{8}\bar{9}}
\tdr_{\bar{5}7\bar{9}}\tdr_{35\bar{9}}
\tdr_{26\bar{9}}
\tdr_{\bar{6}7\bar{8}}
\tdr_{\bar{2}5\bar{8}}
\tdr_{234}
\tdr_{16\bar{8}}
\tdr_{1\bar{2}4}
\underline{
\tdr_{2\bar{5}8}
\tdr_{1\bar{6}8}
\tdr_{12\bar{4}}
\tdr_{\bar{4}\bar{5}\bar{6}}
}
\tdr_{\bar{2}3\bar{4}}
\tdr_{\bar{2}\bar{6}9}
\tdr_{3\bar{5}9}
\tdr_{456}
\tdr_{678}
\tdr_{579}
\tdr_{489}
\\
&=
\underline{\tdr_{\bar{4}\bar{8}\bar{9}}}%
{\ }
\underline{\tdr_{\bar{5}7\bar{9}}}%
\tdr_{35\bar{9}}
\tdr_{26\bar{9}}
\underline{\tdr_{\bar{6}7\bar{8}}}%
\tdr_{\bar{2}5\bar{8}}
\tdr_{234}
\tdr_{16\bar{8}}
\tdr_{1\bar{2}4}
\underline{\tdr_{\bar{4}\bar{5}\bar{6}}}%
\tdr_{12\bar{4}}
\tdr_{1\bar{6}8}
\tdr_{2\bar{5}8}
\tdr_{\bar{2}3\bar{4}}
\tdr_{\bar{2}\bar{6}9}
\tdr_{3\bar{5}9}
\tdr_{456}
\tdr_{678}
\tdr_{579}
\tdr_{489}
\\
&=
\underline{
\tdr_{\bar{4}\bar{8}\bar{9}}
\tdr_{\bar{5}7\bar{9}}
\tdr_{\bar{6}7\bar{8}}
\tdr_{\bar{4}\bar{5}\bar{6}}
}
\tdr_{35\bar{9}}
\tdr_{26\bar{9}}
\tdr_{\bar{2}5\bar{8}}
\tdr_{234}
\tdr_{16\bar{8}}
\tdr_{1\bar{2}4}
\tdr_{12\bar{4}}
\tdr_{1\bar{6}8}
\tdr_{2\bar{5}8}
\tdr_{\bar{2}3\bar{4}}
\tdr_{\bar{2}\bar{6}9}
\tdr_{3\bar{5}9}
\tdr_{456}
\tdr_{678}
\tdr_{579}
\tdr_{489}
\\
&=
\tdr_{\bar{4}\bar{5}\bar{6}}
\tdr_{\bar{6}7\bar{8}}
\tdr_{\bar{5}7\bar{9}}
\tdr_{\bar{4}\bar{8}\bar{9}}
\tdr_{35\bar{9}}
\tdr_{26\bar{9}}
\underline{\tdr_{\bar{2}5\bar{8}}}%
\tdr_{234}
\underline{\tdr_{16\bar{8}}}%
{\ }
\underline{\tdr_{1\bar{2}4}}%
\tdr_{12\bar{4}}
\tdr_{1\bar{6}8}
\tdr_{2\bar{5}8}
\tdr_{\bar{2}3\bar{4}}
\tdr_{\bar{2}\bar{6}9}
\tdr_{3\bar{5}9}
\underline{\tdr_{456}}%
\tdr_{678}
\tdr_{579}
\tdr_{489}
\\
&=
\tdr_{\bar{4}\bar{5}\bar{6}}
\tdr_{\bar{6}7\bar{8}}
\tdr_{\bar{5}7\bar{9}}
\tdr_{\bar{4}\bar{8}\bar{9}}
\tdr_{35\bar{9}}
\tdr_{26\bar{9}}
\tdr_{234}
\underline{
\tdr_{\bar{2}5\bar{8}}
\tdr_{16\bar{8}}
\tdr_{1\bar{2}4}
\tdr_{456}
}
\tdr_{12\bar{4}}
\tdr_{1\bar{6}8}
\tdr_{2\bar{5}8}
\tdr_{\bar{2}3\bar{4}}
\tdr_{\bar{2}\bar{6}9}
\tdr_{3\bar{5}9}
\tdr_{678}
\tdr_{579}
\tdr_{489}
\\
&=
\tdr_{\bar{4}\bar{5}\bar{6}}
\tdr_{\bar{6}7\bar{8}}
\tdr_{\bar{5}7\bar{9}}
\tdr_{\bar{4}\bar{8}\bar{9}}
\underline{
\tdr_{35\bar{9}}
\tdr_{26\bar{9}}
\tdr_{234}
\tdr_{456}
}
\tdr_{1\bar{2}4}
\tdr_{16\bar{8}}
\tdr_{\bar{2}5\bar{8}}
\tdr_{12\bar{4}}
\tdr_{1\bar{6}8}
\tdr_{2\bar{5}8}
\tdr_{\bar{2}3\bar{4}}
\tdr_{\bar{2}\bar{6}9}
\tdr_{3\bar{5}9}
\tdr_{678}
\tdr_{579}
\tdr_{489}
\\
&=
\tdr_{\bar{4}\bar{5}\bar{6}}
\tdr_{\bar{6}7\bar{8}}
\tdr_{\bar{5}7\bar{9}}
\underline{\tdr_{\bar{4}\bar{8}\bar{9}}}%
\tdr_{456}
\tdr_{234}
\underline{\tdr_{26\bar{9}}}%
\tdr_{35\bar{9}}
\tdr_{1\bar{2}4}
\underline{\tdr_{16\bar{8}}}%
\tdr_{\bar{2}5\bar{8}}
\underline{\tdr_{12\bar{4}}}%
\tdr_{1\bar{6}8}
\tdr_{2\bar{5}8}
\tdr_{\bar{2}3\bar{4}}
\tdr_{\bar{2}\bar{6}9}
\tdr_{3\bar{5}9}
\tdr_{678}
\tdr_{579}
\tdr_{489}
\\
&=
\tdr_{\bar{4}\bar{5}\bar{6}}
\tdr_{\bar{6}7\bar{8}}
\tdr_{\bar{5}7\bar{9}}
\tdr_{456}
\tdr_{234}
\tdr_{1\bar{2}4}
\underline{
\tdr_{\bar{4}\bar{8}\bar{9}}
\tdr_{26\bar{9}}
\tdr_{16\bar{8}}
\tdr_{12\bar{4}}
}
\tdr_{35\bar{9}}
\tdr_{\bar{2}5\bar{8}}
\tdr_{1\bar{6}8}
\tdr_{2\bar{5}8}
\tdr_{\bar{2}3\bar{4}}
\tdr_{\bar{2}\bar{6}9}
\tdr_{3\bar{5}9}
\tdr_{678}
\tdr_{579}
\tdr_{489}
\\
&=
\tdr_{\bar{4}\bar{5}\bar{6}}
\tdr_{\bar{6}7\bar{8}}
\tdr_{\bar{5}7\bar{9}}
\tdr_{456}
\tdr_{234}
\tdr_{1\bar{2}4}
\tdr_{12\bar{4}}
\tdr_{16\bar{8}}
\tdr_{26\bar{9}}
\underline{\tdr_{\bar{4}\bar{8}\bar{9}}}%
{\ }
\underline{\tdr_{35\bar{9}}}%
{\ }
\underline{\tdr_{\bar{2}5\bar{8}}}%
\tdr_{1\bar{6}8}
\tdr_{2\bar{5}8}
\underline{\tdr_{\bar{2}3\bar{4}}}%
\tdr_{\bar{2}\bar{6}9}
\tdr_{3\bar{5}9}
\tdr_{678}
\tdr_{579}
\tdr_{489}
\\
&=
\tdr_{\bar{4}\bar{5}\bar{6}}
\tdr_{\bar{6}7\bar{8}}
\tdr_{\bar{5}7\bar{9}}
\tdr_{456}
\tdr_{234}
\tdr_{1\bar{2}4}
\tdr_{12\bar{4}}
\tdr_{16\bar{8}}
\tdr_{26\bar{9}}
\underline{
\tdr_{\bar{4}\bar{8}\bar{9}}
\tdr_{35\bar{9}}
\tdr_{\bar{2}5\bar{8}}
\tdr_{\bar{2}3\bar{4}}
}
\tdr_{1\bar{6}8}
\tdr_{2\bar{5}8}
\tdr_{\bar{2}\bar{6}9}
\tdr_{3\bar{5}9}
\tdr_{678}
\tdr_{579}
\tdr_{489}
\\
&=
\tdr_{\bar{4}\bar{5}\bar{6}}
\tdr_{\bar{6}7\bar{8}}
\underline{\tdr_{\bar{5}7\bar{9}}}%
\tdr_{456}
\tdr_{234}
\tdr_{1\bar{2}4}
\tdr_{12\bar{4}}
\tdr_{16\bar{8}}
\underline{\tdr_{26\bar{9}}}%
\tdr_{\bar{2}3\bar{4}}
\tdr_{\bar{2}5\bar{8}}
\tdr_{35\bar{9}}
\tdr_{\bar{4}\bar{8}\bar{9}}
\tdr_{1\bar{6}8}
\underline{\tdr_{2\bar{5}8}}%
\tdr_{\bar{2}\bar{6}9}
\tdr_{3\bar{5}9}
\underline{\tdr_{678}}%
\tdr_{579}
\tdr_{489}
\\
&=
\tdr_{\bar{4}\bar{5}\bar{6}}
\tdr_{\bar{6}7\bar{8}}
\tdr_{456}
\tdr_{234}
\tdr_{1\bar{2}4}
\tdr_{12\bar{4}}
\tdr_{16\bar{8}}
\tdr_{\bar{2}3\bar{4}}
\tdr_{\bar{2}5\bar{8}}
\tdr_{1\bar{6}8}
\underline{
\tdr_{\bar{5}7\bar{9}}
\tdr_{26\bar{9}}
\tdr_{2\bar{5}8}
\tdr_{678}
}
\tdr_{35\bar{9}}
\tdr_{\bar{4}\bar{8}\bar{9}}
\tdr_{\bar{2}\bar{6}9}
\tdr_{3\bar{5}9}
\tdr_{579}
\tdr_{489}
\\
&=
\tdr_{\bar{4}\bar{5}\bar{6}}
\tdr_{\bar{6}7\bar{8}}
\tdr_{456}
\tdr_{234}
\tdr_{1\bar{2}4}
\tdr_{12\bar{4}}
\tdr_{16\bar{8}}
\tdr_{\bar{2}3\bar{4}}
\tdr_{\bar{2}5\bar{8}}
\tdr_{1\bar{6}8}
\tdr_{678}
\tdr_{2\bar{5}8}
\tdr_{26\bar{9}}
\underline{\tdr_{\bar{5}7\bar{9}}}%
{\ }
\underline{\tdr_{35\bar{9}}}%
\tdr_{\bar{4}\bar{8}\bar{9}}
\tdr_{\bar{2}\bar{6}9}
\underline{\tdr_{3\bar{5}9}}%
{\ }
\underline{\tdr_{579}}%
\tdr_{489}
\\
&=
\tdr_{456}
\tdr_{\bar{4}\bar{5}\bar{6}}
\tdr_{234}\tdr_{1\bar{2}4}\tdr_{12\bar{4}}\tdr_{\bar{2}3\bar{4}}
\tdr_{\bar{6}7\bar{8}}\tdr_{16\bar{8}}\tdr_{1\bar{6}8}\tdr_{678}
\tdr_{\bar{2}5\bar{8}}
\tdr_{2\bar{5}8}
\tdr_{\bar{2}\bar{6}9}
\tdr_{26\bar{9}}
\underline{
\tdr_{\bar{5}7\bar{9}}\tdr_{35\bar{9}}\tdr_{3\bar{5}9}\tdr_{579}
}
\tdr_{\bar{4}\bar{8}\bar{9}}
\tdr_{489}
\\
&=
\tdr_{456}
\tdr_{\bar{4}\bar{5}\bar{6}}
\tdr_{234}\tdr_{1\bar{2}4}\tdr_{12\bar{4}}\tdr_{\bar{2}3\bar{4}}
\tdr_{\bar{6}7\bar{8}}\tdr_{16\bar{8}}\tdr_{1\bar{6}8}\tdr_{678}
\tdr_{\bar{2}5\bar{8}}
\tdr_{2\bar{5}8}
\tdr_{\bar{2}\bar{6}9}
\tdr_{26\bar{9}}
\tdr_{579}\tdr_{3\bar{5}9}\tdr_{35\bar{9}}\tdr_{\bar{5}7\bar{9}}
\tdr_{\bar{4}\bar{8}\bar{9}}
\tdr_{489}
.
\end{align}
\end{landscape}
\begin{landscape}
\section{Proof of Lemma \ref{R20 super lemma}}\label{app R20 super}
\begin{align}
&
\tdmC_{489}
\tdmC_{\bar{4}\bar{8}\bar{9}}
\underline{
\tdnC_{579}\tdmC_{3\bar{5}9}\tdmC_{35\bar{9}}\tdnC_{\bar{5}7\bar{9}}
}
\tdmC_{26\bar{9}}
\tdmC_{\bar{2}\bar{6}9}
\tdmC_{2\bar{5}8}
\tdmC_{\bar{2}5\bar{8}}
\tdnC_{\bar{6}7\bar{8}}
\tdmC_{16\bar{8}}\tdmC_{1\bar{6}8}\tdnC_{678}
\tdrC_{234}\tdrC_{1\bar{2}4}\tdrC_{12\bar{4}}\tdrC_{\bar{2}3\bar{4}}
\tdmC_{\bar{4}\bar{5}\bar{6}}
\tdmC_{456}
\\
&=
\tdmC_{489}
\tdmC_{\bar{4}\bar{8}\bar{9}}
\tdnC_{\bar{5}7\bar{9}}\tdmC_{35\bar{9}}\tdmC_{3\bar{5}9}
\underline{\tdnC_{579}}%
\tdmC_{26\bar{9}}
\underline{\tdmC_{\bar{2}\bar{6}9}}%
\tdmC_{2\bar{5}8}
\underline{\tdmC_{\bar{2}5\bar{8}}}%
{\ }
\underline{\tdnC_{\bar{6}7\bar{8}}}%
\tdmC_{16\bar{8}}\tdmC_{1\bar{6}8}\tdnC_{678}
\tdrC_{234}\tdrC_{1\bar{2}4}\tdrC_{12\bar{4}}\tdrC_{\bar{2}3\bar{4}}
\tdmC_{\bar{4}\bar{5}\bar{6}}
\tdmC_{456}
\\
&=
\tdmC_{489}
\tdmC_{\bar{4}\bar{8}\bar{9}}
\tdnC_{\bar{5}7\bar{9}}\tdmC_{35\bar{9}}\tdmC_{3\bar{5}9}
\tdmC_{26\bar{9}}
\tdmC_{2\bar{5}8}
\underline{
\tdnC_{579}
\tdmC_{\bar{2}\bar{6}9}
\tdmC_{\bar{2}5\bar{8}}
\tdnC_{\bar{6}7\bar{8}}
}%
\tdmC_{16\bar{8}}\tdmC_{1\bar{6}8}\tdnC_{678}
\tdrC_{234}\tdrC_{1\bar{2}4}\tdrC_{12\bar{4}}\tdrC_{\bar{2}3\bar{4}}
\tdmC_{\bar{4}\bar{5}\bar{6}}
\tdmC_{456}
\\
&=
\underline{\tdmC_{489}}%
\tdmC_{\bar{4}\bar{8}\bar{9}}
\tdnC_{\bar{5}7\bar{9}}\tdmC_{35\bar{9}}
\underline{\tdmC_{3\bar{5}9}}%
\tdmC_{26\bar{9}}
\underline{\tdmC_{2\bar{5}8}}%
\tdnC_{\bar{6}7\bar{8}}
\tdmC_{\bar{2}5\bar{8}}
\tdmC_{\bar{2}\bar{6}9}
\tdnC_{579}
\tdmC_{16\bar{8}}\tdmC_{1\bar{6}8}\tdnC_{678}
\underline{\tdrC_{234}}%
\tdrC_{1\bar{2}4}\tdrC_{12\bar{4}}\tdrC_{\bar{2}3\bar{4}}
\tdmC_{\bar{4}\bar{5}\bar{6}}
\tdmC_{456}
\\
&=
\tdmC_{\bar{4}\bar{8}\bar{9}}
\tdnC_{\bar{5}7\bar{9}}\tdmC_{35\bar{9}}
\tdmC_{26\bar{9}}
\tdnC_{\bar{6}7\bar{8}}
\tdmC_{\bar{2}5\bar{8}}
\underline{
\tdmC_{489}
\tdmC_{3\bar{5}9}
\tdmC_{2\bar{5}8}
\tdrC_{234}
}
\tdmC_{\bar{2}\bar{6}9}
\tdnC_{579}
\tdmC_{16\bar{8}}\tdmC_{1\bar{6}8}\tdnC_{678}
\tdrC_{1\bar{2}4}\tdrC_{12\bar{4}}\tdrC_{\bar{2}3\bar{4}}
\tdmC_{\bar{4}\bar{5}\bar{6}}
\tdmC_{456}
\\
&=
\tdmC_{\bar{4}\bar{8}\bar{9}}
\tdnC_{\bar{5}7\bar{9}}\tdmC_{35\bar{9}}
\tdmC_{26\bar{9}}
\tdnC_{\bar{6}7\bar{8}}
\tdmC_{\bar{2}5\bar{8}}
\tdrC_{234}
\tdmC_{2\bar{5}8}
\tdmC_{3\bar{5}9}
\underline{\tdmC_{489}}%
{\ }
\underline{\tdmC_{\bar{2}\bar{6}9}}%
\tdnC_{579}
\tdmC_{16\bar{8}}
\underline{\tdmC_{1\bar{6}8}}%
\tdnC_{678}
\underline{\tdrC_{1\bar{2}4}}%
\tdrC_{12\bar{4}}\tdrC_{\bar{2}3\bar{4}}
\tdmC_{\bar{4}\bar{5}\bar{6}}
\tdmC_{456}
\\
&=
\tdmC_{\bar{4}\bar{8}\bar{9}}
\tdnC_{\bar{5}7\bar{9}}\tdmC_{35\bar{9}}
\tdmC_{26\bar{9}}
\tdnC_{\bar{6}7\bar{8}}
\tdmC_{\bar{2}5\bar{8}}
\tdrC_{234}
\tdmC_{2\bar{5}8}
\tdmC_{3\bar{5}9}
\tdmC_{16\bar{8}}
\underline{
\tdmC_{489}
\tdmC_{\bar{2}\bar{6}9}
\tdmC_{1\bar{6}8}
\tdrC_{1\bar{2}4}
}
\tdnC_{579}
\tdnC_{678}
\tdrC_{12\bar{4}}\tdrC_{\bar{2}3\bar{4}}
\tdmC_{\bar{4}\bar{5}\bar{6}}
\tdmC_{456}
\\
&=
\tdmC_{\bar{4}\bar{8}\bar{9}}
\tdnC_{\bar{5}7\bar{9}}\tdmC_{35\bar{9}}
\tdmC_{26\bar{9}}
\tdnC_{\bar{6}7\bar{8}}
\tdmC_{\bar{2}5\bar{8}}
\tdrC_{234}
\tdmC_{2\bar{5}8}
\tdmC_{3\bar{5}9}
\tdmC_{16\bar{8}}
\tdrC_{1\bar{2}4}
\tdmC_{1\bar{6}8}
\tdmC_{\bar{2}\bar{6}9}
\underline{\tdmC_{489}}%
{\ }
\underline{\tdnC_{579}}%
{\ }
\underline{\tdnC_{678}}%
\tdrC_{12\bar{4}}\tdrC_{\bar{2}3\bar{4}}
\tdmC_{\bar{4}\bar{5}\bar{6}}
\underline{\tdmC_{456}}
\\
&=
\tdmC_{\bar{4}\bar{8}\bar{9}}
\tdnC_{\bar{5}7\bar{9}}\tdmC_{35\bar{9}}
\tdmC_{26\bar{9}}
\tdnC_{\bar{6}7\bar{8}}
\tdmC_{\bar{2}5\bar{8}}
\tdrC_{234}
\tdmC_{2\bar{5}8}
\tdmC_{3\bar{5}9}
\tdmC_{16\bar{8}}
\tdrC_{1\bar{2}4}
\tdmC_{1\bar{6}8}
\tdmC_{\bar{2}\bar{6}9}
\tdrC_{12\bar{4}}\tdrC_{\bar{2}3\bar{4}}
\tdmC_{\bar{4}\bar{5}\bar{6}}
\underline{
\tdmC_{489}
\tdnC_{579}
\tdnC_{678}
\tdmC_{456}
}
\\
&=
\tdmC_{\bar{4}\bar{8}\bar{9}}
\tdnC_{\bar{5}7\bar{9}}\tdmC_{35\bar{9}}
\tdmC_{26\bar{9}}
\tdnC_{\bar{6}7\bar{8}}
\tdmC_{\bar{2}5\bar{8}}
\tdrC_{234}
\tdmC_{2\bar{5}8}
\underline{\tdmC_{3\bar{5}9}}%
\tdmC_{16\bar{8}}
\tdrC_{1\bar{2}4}
\tdmC_{1\bar{6}8}
\underline{\tdmC_{\bar{2}\bar{6}9}}%
\tdrC_{12\bar{4}}
\underline{\tdrC_{\bar{2}3\bar{4}}}%
{\ }
\underline{\tdmC_{\bar{4}\bar{5}\bar{6}}}%
\tdmC_{456}
\tdnC_{678}
\tdnC_{579}
\tdmC_{489}
\\
&=
\tdmC_{\bar{4}\bar{8}\bar{9}}
\tdnC_{\bar{5}7\bar{9}}\tdmC_{35\bar{9}}
\tdmC_{26\bar{9}}
\tdnC_{\bar{6}7\bar{8}}
\tdmC_{\bar{2}5\bar{8}}
\tdrC_{234}
\tdmC_{2\bar{5}8}
\tdmC_{16\bar{8}}
\tdrC_{1\bar{2}4}
\tdmC_{1\bar{6}8}
\tdrC_{12\bar{4}}
\underline{
\tdmC_{3\bar{5}9}
\tdmC_{\bar{2}\bar{6}9}
\tdrC_{\bar{2}3\bar{4}}
\tdmC_{\bar{4}\bar{5}\bar{6}}
}
\tdmC_{456}
\tdnC_{678}
\tdnC_{579}
\tdmC_{489}
\\
&=
\tdmC_{\bar{4}\bar{8}\bar{9}}
\tdnC_{\bar{5}7\bar{9}}\tdmC_{35\bar{9}}
\tdmC_{26\bar{9}}
\tdnC_{\bar{6}7\bar{8}}
\tdmC_{\bar{2}5\bar{8}}
\tdrC_{234}
\underline{\tdmC_{2\bar{5}8}}%
\tdmC_{16\bar{8}}
\tdrC_{1\bar{2}4}
\underline{\tdmC_{1\bar{6}8}}%
{\ }
\underline{\tdrC_{12\bar{4}}}%
{\ }
\underline{\tdmC_{\bar{4}\bar{5}\bar{6}}}%
\tdrC_{\bar{2}3\bar{4}}
\tdmC_{\bar{2}\bar{6}9}
\tdmC_{3\bar{5}9}
\tdmC_{456}
\tdnC_{678}
\tdnC_{579}
\tdmC_{489}
\\
&=
\tdmC_{\bar{4}\bar{8}\bar{9}}
\tdnC_{\bar{5}7\bar{9}}\tdmC_{35\bar{9}}
\tdmC_{26\bar{9}}
\tdnC_{\bar{6}7\bar{8}}
\tdmC_{\bar{2}5\bar{8}}
\tdrC_{234}
\tdmC_{16\bar{8}}
\tdrC_{1\bar{2}4}
\underline{
\tdmC_{2\bar{5}8}
\tdmC_{1\bar{6}8}
\tdrC_{12\bar{4}}
\tdmC_{\bar{4}\bar{5}\bar{6}}
}
\tdrC_{\bar{2}3\bar{4}}
\tdmC_{\bar{2}\bar{6}9}
\tdmC_{3\bar{5}9}
\tdmC_{456}
\tdnC_{678}
\tdnC_{579}
\tdmC_{489}
\\
&=
\underline{\tdmC_{\bar{4}\bar{8}\bar{9}}}%
{\ }
\underline{\tdnC_{\bar{5}7\bar{9}}}%
\tdmC_{35\bar{9}}
\tdmC_{26\bar{9}}
\underline{\tdnC_{\bar{6}7\bar{8}}}%
\tdmC_{\bar{2}5\bar{8}}
\tdrC_{234}
\tdmC_{16\bar{8}}
\tdrC_{1\bar{2}4}
\underline{\tdmC_{\bar{4}\bar{5}\bar{6}}}%
\tdrC_{12\bar{4}}
\tdmC_{1\bar{6}8}
\tdmC_{2\bar{5}8}
\tdrC_{\bar{2}3\bar{4}}
\tdmC_{\bar{2}\bar{6}9}
\tdmC_{3\bar{5}9}
\tdmC_{456}
\tdnC_{678}
\tdnC_{579}
\tdmC_{489}
\\
&=
\underline{
\tdmC_{\bar{4}\bar{8}\bar{9}}
\tdnC_{\bar{5}7\bar{9}}
\tdnC_{\bar{6}7\bar{8}}
\tdmC_{\bar{4}\bar{5}\bar{6}}
}
\tdmC_{35\bar{9}}
\tdmC_{26\bar{9}}
\tdmC_{\bar{2}5\bar{8}}
\tdrC_{234}
\tdmC_{16\bar{8}}
\tdrC_{1\bar{2}4}
\tdrC_{12\bar{4}}
\tdmC_{1\bar{6}8}
\tdmC_{2\bar{5}8}
\tdrC_{\bar{2}3\bar{4}}
\tdmC_{\bar{2}\bar{6}9}
\tdmC_{3\bar{5}9}
\tdmC_{456}
\tdnC_{678}
\tdnC_{579}
\tdmC_{489}
\\
&=
\tdmC_{\bar{4}\bar{5}\bar{6}}
\tdnC_{\bar{6}7\bar{8}}
\tdnC_{\bar{5}7\bar{9}}
\tdmC_{\bar{4}\bar{8}\bar{9}}
\tdmC_{35\bar{9}}
\tdmC_{26\bar{9}}
\underline{\tdmC_{\bar{2}5\bar{8}}}%
\tdrC_{234}
\underline{\tdmC_{16\bar{8}}}%
{\ }
\underline{\tdrC_{1\bar{2}4}}%
\tdrC_{12\bar{4}}
\tdmC_{1\bar{6}8}
\tdmC_{2\bar{5}8}
\tdrC_{\bar{2}3\bar{4}}
\tdmC_{\bar{2}\bar{6}9}
\tdmC_{3\bar{5}9}
\underline{\tdmC_{456}}%
\tdnC_{678}
\tdnC_{579}
\tdmC_{489}
\\
&=
\tdmC_{\bar{4}\bar{5}\bar{6}}
\tdnC_{\bar{6}7\bar{8}}
\tdnC_{\bar{5}7\bar{9}}
\tdmC_{\bar{4}\bar{8}\bar{9}}
\tdmC_{35\bar{9}}
\tdmC_{26\bar{9}}
\tdrC_{234}
\underline{
\tdmC_{\bar{2}5\bar{8}}
\tdmC_{16\bar{8}}
\tdrC_{1\bar{2}4}
\tdmC_{456}
}
\tdrC_{12\bar{4}}
\tdmC_{1\bar{6}8}
\tdmC_{2\bar{5}8}
\tdrC_{\bar{2}3\bar{4}}
\tdmC_{\bar{2}\bar{6}9}
\tdmC_{3\bar{5}9}
\tdnC_{678}
\tdnC_{579}
\tdmC_{489}
\\
&=
\tdmC_{\bar{4}\bar{5}\bar{6}}
\tdnC_{\bar{6}7\bar{8}}
\tdnC_{\bar{5}7\bar{9}}
\tdmC_{\bar{4}\bar{8}\bar{9}}
\underline{
\tdmC_{35\bar{9}}
\tdmC_{26\bar{9}}
\tdrC_{234}
\tdmC_{456}
}
\tdrC_{1\bar{2}4}
\tdmC_{16\bar{8}}
\tdmC_{\bar{2}5\bar{8}}
\tdrC_{12\bar{4}}
\tdmC_{1\bar{6}8}
\tdmC_{2\bar{5}8}
\tdrC_{\bar{2}3\bar{4}}
\tdmC_{\bar{2}\bar{6}9}
\tdmC_{3\bar{5}9}
\tdnC_{678}
\tdnC_{579}
\tdmC_{489}
\\
&=
\tdmC_{\bar{4}\bar{5}\bar{6}}
\tdnC_{\bar{6}7\bar{8}}
\tdnC_{\bar{5}7\bar{9}}
\underline{\tdmC_{\bar{4}\bar{8}\bar{9}}}%
\tdmC_{456}
\tdrC_{234}
\underline{\tdmC_{26\bar{9}}}%
\tdmC_{35\bar{9}}
\tdrC_{1\bar{2}4}
\underline{\tdmC_{16\bar{8}}}%
\tdmC_{\bar{2}5\bar{8}}
\underline{\tdrC_{12\bar{4}}}%
\tdmC_{1\bar{6}8}
\tdmC_{2\bar{5}8}
\tdrC_{\bar{2}3\bar{4}}
\tdmC_{\bar{2}\bar{6}9}
\tdmC_{3\bar{5}9}
\tdnC_{678}
\tdnC_{579}
\tdmC_{489}
\\
&=
\tdmC_{\bar{4}\bar{5}\bar{6}}
\tdnC_{\bar{6}7\bar{8}}
\tdnC_{\bar{5}7\bar{9}}
\tdmC_{456}
\tdrC_{234}
\tdrC_{1\bar{2}4}
\underline{
\tdmC_{\bar{4}\bar{8}\bar{9}}
\tdmC_{26\bar{9}}
\tdmC_{16\bar{8}}
\tdrC_{12\bar{4}}
}
\tdmC_{35\bar{9}}
\tdmC_{\bar{2}5\bar{8}}
\tdmC_{1\bar{6}8}
\tdmC_{2\bar{5}8}
\tdrC_{\bar{2}3\bar{4}}
\tdmC_{\bar{2}\bar{6}9}
\tdmC_{3\bar{5}9}
\tdnC_{678}
\tdnC_{579}
\tdmC_{489}
\\
&=
\tdmC_{\bar{4}\bar{5}\bar{6}}
\tdnC_{\bar{6}7\bar{8}}
\tdnC_{\bar{5}7\bar{9}}
\tdmC_{456}
\tdrC_{234}
\tdrC_{1\bar{2}4}
\tdrC_{12\bar{4}}
\tdmC_{16\bar{8}}
\tdmC_{26\bar{9}}
\underline{\tdmC_{\bar{4}\bar{8}\bar{9}}}%
{\ }
\underline{\tdmC_{35\bar{9}}}%
{\ }
\underline{\tdmC_{\bar{2}5\bar{8}}}%
\tdmC_{1\bar{6}8}
\tdmC_{2\bar{5}8}
\underline{\tdrC_{\bar{2}3\bar{4}}}%
\tdmC_{\bar{2}\bar{6}9}
\tdmC_{3\bar{5}9}
\tdnC_{678}
\tdnC_{579}
\tdmC_{489}
\\
&=
\tdmC_{\bar{4}\bar{5}\bar{6}}
\tdnC_{\bar{6}7\bar{8}}
\tdnC_{\bar{5}7\bar{9}}
\tdmC_{456}
\tdrC_{234}
\tdrC_{1\bar{2}4}
\tdrC_{12\bar{4}}
\tdmC_{16\bar{8}}
\tdmC_{26\bar{9}}
\underline{
\tdmC_{\bar{4}\bar{8}\bar{9}}
\tdmC_{35\bar{9}}
\tdmC_{\bar{2}5\bar{8}}
\tdrC_{\bar{2}3\bar{4}}
}
\tdmC_{1\bar{6}8}
\tdmC_{2\bar{5}8}
\tdmC_{\bar{2}\bar{6}9}
\tdmC_{3\bar{5}9}
\tdnC_{678}
\tdnC_{579}
\tdmC_{489}
\\
&=
\tdmC_{\bar{4}\bar{5}\bar{6}}
\tdnC_{\bar{6}7\bar{8}}
\underline{\tdnC_{\bar{5}7\bar{9}}}%
\tdmC_{456}
\tdrC_{234}
\tdrC_{1\bar{2}4}
\tdrC_{12\bar{4}}
\tdmC_{16\bar{8}}
\underline{\tdmC_{26\bar{9}}}%
\tdrC_{\bar{2}3\bar{4}}
\tdmC_{\bar{2}5\bar{8}}
\tdmC_{35\bar{9}}
\tdmC_{\bar{4}\bar{8}\bar{9}}
\tdmC_{1\bar{6}8}
\underline{\tdmC_{2\bar{5}8}}%
\tdmC_{\bar{2}\bar{6}9}
\tdmC_{3\bar{5}9}
\underline{\tdnC_{678}}%
\tdnC_{579}
\tdmC_{489}
\\
&=
\tdmC_{\bar{4}\bar{5}\bar{6}}
\tdnC_{\bar{6}7\bar{8}}
\tdmC_{456}
\tdrC_{234}
\tdrC_{1\bar{2}4}
\tdrC_{12\bar{4}}
\tdmC_{16\bar{8}}
\tdrC_{\bar{2}3\bar{4}}
\tdmC_{\bar{2}5\bar{8}}
\tdmC_{1\bar{6}8}
\underline{
\tdnC_{\bar{5}7\bar{9}}
\tdmC_{26\bar{9}}
\tdmC_{2\bar{5}8}
\tdnC_{678}
}
\tdmC_{35\bar{9}}
\tdmC_{\bar{4}\bar{8}\bar{9}}
\tdmC_{\bar{2}\bar{6}9}
\tdmC_{3\bar{5}9}
\tdnC_{579}
\tdmC_{489}
\\
&=
\tdmC_{\bar{4}\bar{5}\bar{6}}
\tdnC_{\bar{6}7\bar{8}}
\tdmC_{456}
\tdrC_{234}
\tdrC_{1\bar{2}4}
\tdrC_{12\bar{4}}
\tdmC_{16\bar{8}}
\tdrC_{\bar{2}3\bar{4}}
\tdmC_{\bar{2}5\bar{8}}
\tdmC_{1\bar{6}8}
\tdnC_{678}
\tdmC_{2\bar{5}8}
\tdmC_{26\bar{9}}
\underline{\tdnC_{\bar{5}7\bar{9}}}%
{\ }
\underline{\tdmC_{35\bar{9}}}%
\tdmC_{\bar{4}\bar{8}\bar{9}}
\tdmC_{\bar{2}\bar{6}9}
\underline{\tdmC_{3\bar{5}9}}%
{\ }
\underline{\tdnC_{579}}%
\tdmC_{489}
\\
&=
\tdmC_{456}
\tdmC_{\bar{4}\bar{5}\bar{6}}
\tdrC_{234}\tdrC_{1\bar{2}4}\tdrC_{12\bar{4}}\tdrC_{\bar{2}3\bar{4}}
\tdnC_{\bar{6}7\bar{8}}\tdmC_{16\bar{8}}\tdmC_{1\bar{6}8}\tdnC_{678}
\tdmC_{\bar{2}5\bar{8}}
\tdmC_{2\bar{5}8}
\tdmC_{\bar{2}\bar{6}9}
\tdmC_{26\bar{9}}
\underline{
\tdnC_{\bar{5}7\bar{9}}\tdmC_{35\bar{9}}\tdmC_{3\bar{5}9}\tdnC_{579}
}
\tdmC_{\bar{4}\bar{8}\bar{9}}
\tdmC_{489}
\\
&=
\tdmC_{456}
\tdmC_{\bar{4}\bar{5}\bar{6}}
\tdrC_{234}\tdrC_{1\bar{2}4}\tdrC_{12\bar{4}}\tdrC_{\bar{2}3\bar{4}}
\tdnC_{\bar{6}7\bar{8}}\tdmC_{16\bar{8}}\tdmC_{1\bar{6}8}\tdnC_{678}
\tdmC_{\bar{2}5\bar{8}}
\tdmC_{2\bar{5}8}
\tdmC_{\bar{2}\bar{6}9}
\tdmC_{26\bar{9}}
\tdnC_{579}\tdmC_{3\bar{5}9}\tdmC_{35\bar{9}}\tdnC_{\bar{5}7\bar{9}}
\tdmC_{\bar{4}\bar{8}\bar{9}}
\tdmC_{489}
.
\end{align}
\end{landscape}
\end{document}